\documentclass[journal, 10pt]{IEEEtran}

\usepackage{marvosym}
\usepackage{orcidlink}

\usepackage{IEEE-nlwmode}

\begin{document}
\begin{CJK}{UTF8}{gbsn}

%
\title{\textit{BoolSkeleton}: Boolean Network Skeletonization \\ via Homogeneous Pattern Reduction}

\author{
\IEEEauthorblockN{Liwei~Ni$^{\star,}$\textsuperscript{1,2,5},}
\and
\IEEEauthorblockN{Jiaxi~Zhang\textsuperscript{3},}
\and
\IEEEauthorblockN{Shenggen~Zheng\textsuperscript{4},}
\and
\IEEEauthorblockN{Junfeng~Liu\textsuperscript{2},}
\and
\IEEEauthorblockN{Xingyu~Meng\textsuperscript{2},}
\and \\
\IEEEauthorblockN{Biwei~Xie\textsuperscript{1,5},}
\and
\IEEEauthorblockN{Xingquan~Li$^{\text{\Letter},}$\textsuperscript{2},}
 \and
\IEEEauthorblockN{and Huawei~Li\textsuperscript{1,5}}
\and \\

\IEEEauthorblockA{\textsuperscript{1}Institute of Computing Technology, Chinese Academy of Sciences, Beijing, China}  \\
\IEEEauthorblockA{\textsuperscript{2}Pengcheng Laboratory, Shenzhen, China}  \\
\IEEEauthorblockA{\textsuperscript{3}Peking University, Beijing, China}  \\
\IEEEauthorblockA{\textsuperscript{4}Quantum Science Center of Guangdong-Hong Kong-Macao Greater Bay Area, Shenzhen, China}  \\
\IEEEauthorblockA{\textsuperscript{5}University of Chinese Academy of Sciences, Beijing, China}  \\

Emails: \{nlwmode@gmail.com$^{\star}$, lixq01@pcl.ac.cn$^{\text{\Letter}}$\} 
}

\maketitle

\begin{abstract}
Boolean equivalence allows Boolean networks with identical functionality to exhibit diverse graph structures.
This gives more room for exploration in logic optimization, while also posing a challenge for tasks involving consistency between Boolean networks.
To tackle this challenge, we introduce \textit{BoolSkeleton}, a novel Boolean network skeletonization method that improves the consistency and reliability of design-specific evaluations.
\textit{BoolSkeleton} comprises two key steps: preprocessing and reduction. 
In preprocessing, the Boolean network is transformed into a defined Boolean dependency graph, where nodes are assigned the functionality-related status.
Next, the homogeneous and heterogeneous patterns are defined for the node-level pattern reduction step. 
Heterogeneous patterns are preserved to maintain critical functionality-related dependencies, while homogeneous patterns can be reduced.
Parameter $K$ of the pattern further constrains the fanin size of these patterns, enabling fine-tuned control over the granularity of graph reduction.
To validate \textit{BoolSkeleton}’s effectiveness, we conducted four analysis/downstream tasks around the Boolean network: compression analysis, classification, critical path analysis, and timing prediction, demonstrating its robustness across diverse scenarios.
Furthermore, it improves above 55\% in the average accuracy compared to the original Boolean network for the timing prediction task.
These experiments underscore the potential of \textit{BoolSkeleton} to enhance design consistency in logic synthesis.
\end{abstract}


%
\IEEEpeerreviewmaketitle

\section{Introduction}
\label{sec:intro}
\IEEEPARstart{B}{oolean} networks~\cite{multi-level} serve as the intermediate representation in the logic synthesis process~\cite{synthesis_book} within Electronic Design Automation~(EDA), where they can be modeled as a typical computational graph. 
For any given Boolean network, it comprises two key components: ``static'' functionality and ``dynamic'' Directed Acyclic Graph (DAG) structure. 
Here, ``static'' denotes the functionality that remains invariant for a given design, whereas ``dynamic'' reflects the variability of the local DAG structure. 
This variability arises due to the Boolean equivalence theorem~\cite{boolean_equivalence_iccad02}, which posits that Boolean networks with identical functionality can produce diverse DAG structures as a result of logic optimization techniques.
Logic optimization~\cite{synthesis_book} operators aim to reduce the Boolean network's size and depth by the local equivalent replacement techniques~\cite{AIG-2006-rewrite4, refactor_06}, thereby enhancing the efficiency of subsequent EDA steps.
However, this ``dynamic'' flexibility poses challenges for functionality-related graph embedding learning in logic synthesis, such as classification~\cite{classification} and Boolean matching~\cite{boolean_matching_iccad21}.
The variability introduced by optimization complicates the maintenance of consistent representations, creating a tension between structural variance and functional consistency within Boolean networks.

Graph Neural Networks (GNNs) offer a robust framework for learning graph embeddings, effectively extracting consistent features from Boolean networks~\cite{graph_embedding_survey}. 
Several studies have harnessed GNNs for graph embedding tasks in this domain, including DeepGate and its variants~\cite{deepgate, deepgate2, deepgate3}, HOGA~\cite{HOGA}, PolarGate~\cite{polargate}, BoolGebra~\cite{boolgebra}, etc. 
These approaches, however, primarily depend on fine-grained topological structures to represent Boolean networks, placing considerable demands on the expressive capacity of GNNs to capture coarse-grained features. 
While traditional graph coarsening techniques, such as Variations~\cite{skeleton_variations}, Heavy Edge Matching~\cite{skeleton_heavy_edge_matching}, Algebraic Distance~\cite{skeleton_algebraic_distance}, Affinity~\cite{skeleton_affinity}, and Kron Reduction~\cite{skeleton_kron_reduction}, excel at deriving high-dimensional abstractions, their direct application to Boolean networks requires further adaptation. 
This stems from the unique structural and functional properties of Boolean networks, which differ from conventional graphs.
Therefore, there is a critical need to advance these methods by incorporating a global perspective. 
This requires innovative strategies that harmoniously balance local and global feature learning while aligning with the inherent properties.

To address these challenges, we first conduct an in-depth analysis of Boolean networks, uncovering key attributes that define their behavior: Boolean dependency, reachability, reconvergence, and the inherent tension between ``static'' and ``dynamic'' characteristics.
Based on these insights, we propose \textit{BoolSkeleton}, a novel Boolean network skeletonization method that employs homogeneous pattern reduction to balance these attributes while preserving coarse-grained information. 
\textit{BoolSkeleton} consists of two primary phases: preprocessing and reduction.
In preprocessing, the Boolean network is transformed into a Boolean dependency graph, with functionality-related node statuses initialized to reflect their dependencies.
Then, the heterogeneous and homogeneous patterns are defined to facilitate the reduction step.
Heterogeneous patterns can preserve the functionality-dependent structures, while homogeneous patterns enable node reduction to coarsen the graph.
An iterative reduction algorithm, guided by the topological order of the Boolean dependency graph, is then applied to eliminate homogeneous patterns systematically.
We evaluate \textit{BoolSkeleton} across several analysis and downstream tasks: compression, classification, critical path analysis, and timing prediction.
The compression analysis evaluates the network coarsening ratio; the classification analysis validates the consistency; and the critical path analysis task demonstrates its superior profiling capability. Moreover, \textit{BoolSkeleton} achieves over 55\% improvement in average accuracy compared to the original Boolean network in the timing prediction task. These experimental findings highlight the significant potential of \textit{BoolSkeleton} in enhancing the consistency and reliability of Boolean network analysis.
The contributions can be summarised as follows:
\begin{itemize}
    \item We introduce \textit{BoolSkeleton} to coarsen the Boolean network by leveraging the node-level homogeneous pattern reductions. To the best of our knowledge, this is the first work to study the skeleton problem on Boolean networks.
    \item We provide a comprehensive analysis of Boolean networks, identifying key attributes: Boolean dependency, reachability, reconvergence, and the tension between ``static'' and ``dynamic'' attributes.
    \item \textit{BoolSkeleton} can well balance the local functional structure of Boolean networks with the coarse-grained skeleton, overcoming the over-reliance on the fine-grained structure of Boolean networks.
    \item We demonstrate the effectiveness of \textit{BoolSkeleton} by multiple downstream tasks, achieving significant improvements in classification and timing prediction accuracy.
\end{itemize}

The rest of this paper is organized as follows: 
\cref{sec:background} provides an overview of the background and motivation; 
\cref{sec:problem} elaborates on the problem statement; 
\cref{sec:method} details the proposed \textit{BoolSkeleton}; 
\cref{sec:task} presents experimental results for the downstream tasks; 
\cref{sec:discussion} gives the discussion; 
and \cref{sec:conclusion} summarizes the conclusions.

\section{Background and Motivation}
\label{sec:background}

In this section, we will introduce the background of Boolean network and the motivation of this work. 

\vspace{-1em}
\subsection{Background}
\label{sec:background:background}

\subsubsection{Boolean Network}

\begin{figure}[!thb]

    \centering
    \includegraphics[width=0.48\textwidth]{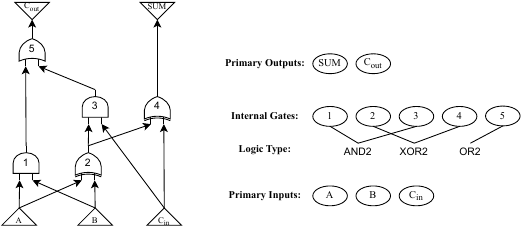}
    \caption{The visualization of a Boolean network for a full adder, where the boolean expression of $SUM$ and $C_{out}$ can be formulated by ``$SUM=C_{in} \oplus (A \oplus B)$, $C_{out}=\big{(}A \wedge B\big{)} \vee \big{(}C_{in} \wedge (A \oplus B)\big{)}$'', respectively.}
    \label{fig:boolean_network}
    \vspace{-1em}
\end{figure}

A Boolean network, denoted as \(\mathcal{C}\), can be modeled as a computational graph which consists of functionality and the gate-based DAG. 
Formally, let \(\mathcal{C} = (\mathcal{V}, \mathcal{E})\), where \(\mathcal{V}\) is the set of nodes and \(\mathcal{E}\) is the set of edges. 
The node set is partitioned as \(\mathcal{V} = \mathcal{V}^{I} \cup \mathcal{V}^{G} \cup \mathcal{V}^{O}\), with \(\mathcal{V}^{I}\) denoting the set of \textit{Primary Input} nodes (PIs), \(\mathcal{V}^{O}\) the set of \textit{Primary Output} nodes (POs), and \(\mathcal{V}^{G}\) the set of internal \textit{logic gates}. 
Edges in \(\mathcal{E}\) refers to the signal propagation: an edge \(v_i \to v_j \in \mathcal{E}\) (where \(v_i, v_j \in \mathcal{V}\)) indicates that \(v_i\) is a \textit{fanin} of \(v_j\), and equivalently, \(v_j\) is a \textit{fanout} of \(v_i\).

For a given node \(v_j\), the set of all its fanin nodes forms its \textit{Transitive FanIn-cone}~(TFI), while the set of all its fanout nodes constitutes its \textit{Transitive FanOut-cone}~(TFO).
For the DAG component of the Boolean network \(\mathcal{C}\) with \(n\) nodes, we define \(\mathbf{A}^{n \times n}\) as the Boolean adjacency matrix, where \(\mathbf{A}_{i,j} = \text{true}\) if there exists an edge \(v_i \to v_j\). 
Additionally, we define \(\mathbf{R}^{n \times n}\) as the reachability matrix, where \(\mathbf{R}_{i,j} = \text{true}\) if there exists a path from \(v_i\) to \(v_j\).
The \textit{depth} of node $v_i \in \mathcal{V}$ is determined using the unit delay model~\cite{unit_delay_iccad96}:
\begin{equation}
\small
\begin{aligned}
\text{depth}(v_i) = 
\begin{cases} 
0, & v_i \in \text{PIs}, \\ 
\max \big{(} \text{depth}_{(\forall v_j \in fanin(v_i))}(v_j)\big{)}  + 1, & \text{otherwise},
\end{cases}
\end{aligned}
\label{eq:unit_delay_model}
\end{equation}
Logic gates within the Boolean network can be constructed from functionally complete sets, such as \{\text{AND2}, \text{INVERTER}\}, \{\text{XOR2}, \text{AND2}, \text{INVERTER}\}, etc.
Furthermore, any superset of a functionally complete set retains functional completeness.
The structure of a Boolean network representing a full adder is illustrated in \cref{fig:boolean_network}.

\subsubsection{Boolean Equivalence}

The Boolean equivalence theory~\cite{boolean_equivalence_iccad02} asserts that Boolean networks with the same functionality can have different graph structures, in other words, Boolean networks with different graph structures can lead to the same functionality.

\begin{figure}[t]
    \centering
    \subfigure[$\mathcal{C}_1:f=(A \wedge B) \wedge (C \wedge D)$]{
    \includegraphics[width=0.2\textwidth]{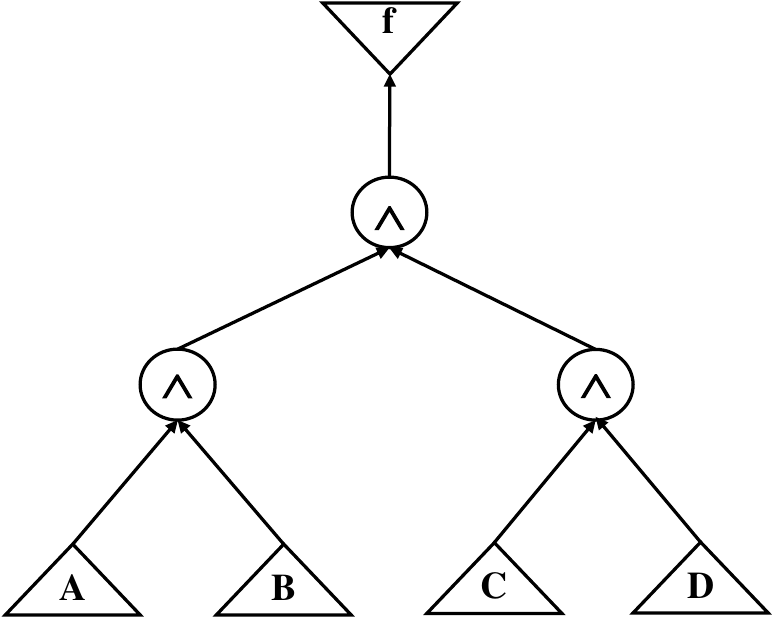}
    }
    \hfill
    \subfigure[$\mathcal{C}_2:g =((A \wedge B) \wedge C) \wedge D$]{
    \includegraphics[width=0.2\textwidth]{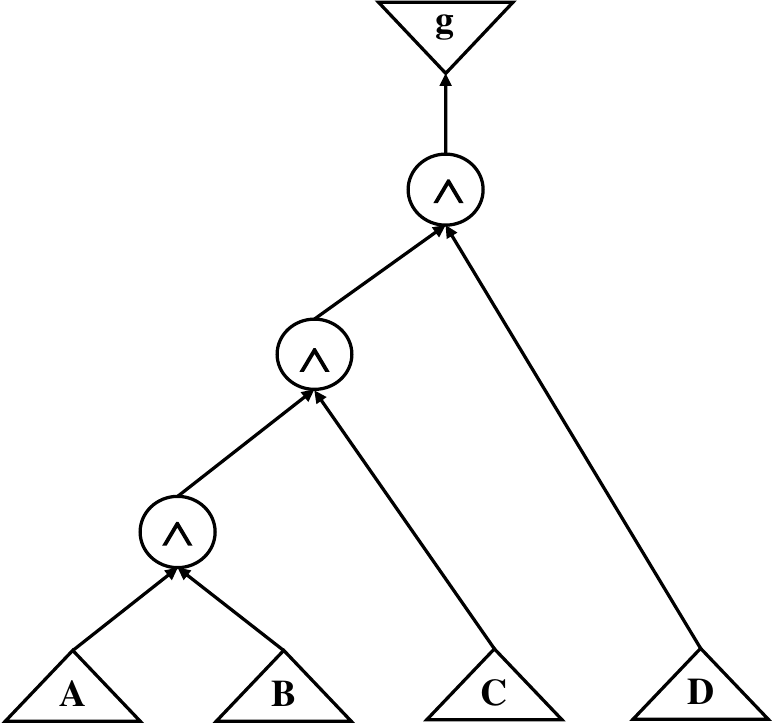}
    }
    \caption{The illustration of the Boolean equivalence: the Boolean expression of $f$ in (a) and $g$ in (b).}
    \label{fig:boolean_equivalence}
    \vspace{-1em}
\end{figure}

\begin{lemma}
\label{lemma:boolean_equivalence}
According to Boolean equivalence, given a pair of Boolean networks ($\mathcal{C}_1$, $\mathcal{C}_2$), we can say that $\mathcal{C}_1$ and $\mathcal{C}_2$ are Boolean equivalent under the following \nlwdel{three}\nlwadd{two} conditions:

\begin{equation}
\small
\begin{aligned}
\text{conditon1:}~&~\big{(}\mathcal{C}_1^{G} \equiv \mathcal{C}_2^{G}\big{)}    &\Rightarrow~~(\mathcal{C}_1 \doteq \mathcal{C}_2); \\
\text{conditon2:}~&~\big{(} (\mathcal{C}_1^{G} \neq   \mathcal{C}_2^{G}) \land (\mathcal{C}_1^{F} \equiv \mathcal{C}_2^{F}) \big{)}  &\Rightarrow~~(\mathcal{C}_1 \doteq \mathcal{C}_2),
\end{aligned}
\end{equation}
where $\doteq$ denotes as the Boolean equivalent operator, $\mathcal{C}^{G}$ refers to the DAG structure of Boolean network $\mathcal{C}$, and $\mathcal{C}^{F}$ refers to the functionality.
\end{lemma}

\cref{lemma:boolean_equivalence} delineates two conditions under which Boolean networks \(\mathcal{C}_1\) and \(\mathcal{C}_2\) are considered equivalent:
\begin{enumerate}
    \item Identical DAG Structures: If \(\mathcal{C}_1^{G} \equiv \mathcal{C}_2^{G}\), identical graph structures imply equivalent functionality, as the same topology consistently with the same node types yields the same Boolean function;
    \item Distinct DAG Structures with Identical Functions: If \(\mathcal{C}_1^{G} \neq \mathcal{C}_2^{G}\) yet \(\mathcal{C}_1^{F} \equiv \mathcal{C}_2^{F}\), equivalence holds despite structural differences, representing a special case of Cond-1.
\end{enumerate}
\cref{fig:boolean_equivalence} illustrates the scenario corresponding to Condition \nlwdel{3}\nlwadd{2} of Boolean equivalence.
For instance, consider two functions \(f\) and \(g\) such that, after logic equivalence transformation, both simplify to \(f = g = A \land B \land C \land D\), thus, \(\mathcal{C}_1 \doteq \mathcal{C}_2\).




\subsection{Motivation}
\subsubsection{Structural Bias}

\begin{figure}[t]
\vspace{-1em}
\centering
\subfigure[size and depth distr.]{
\includegraphics[width=0.2\textwidth]{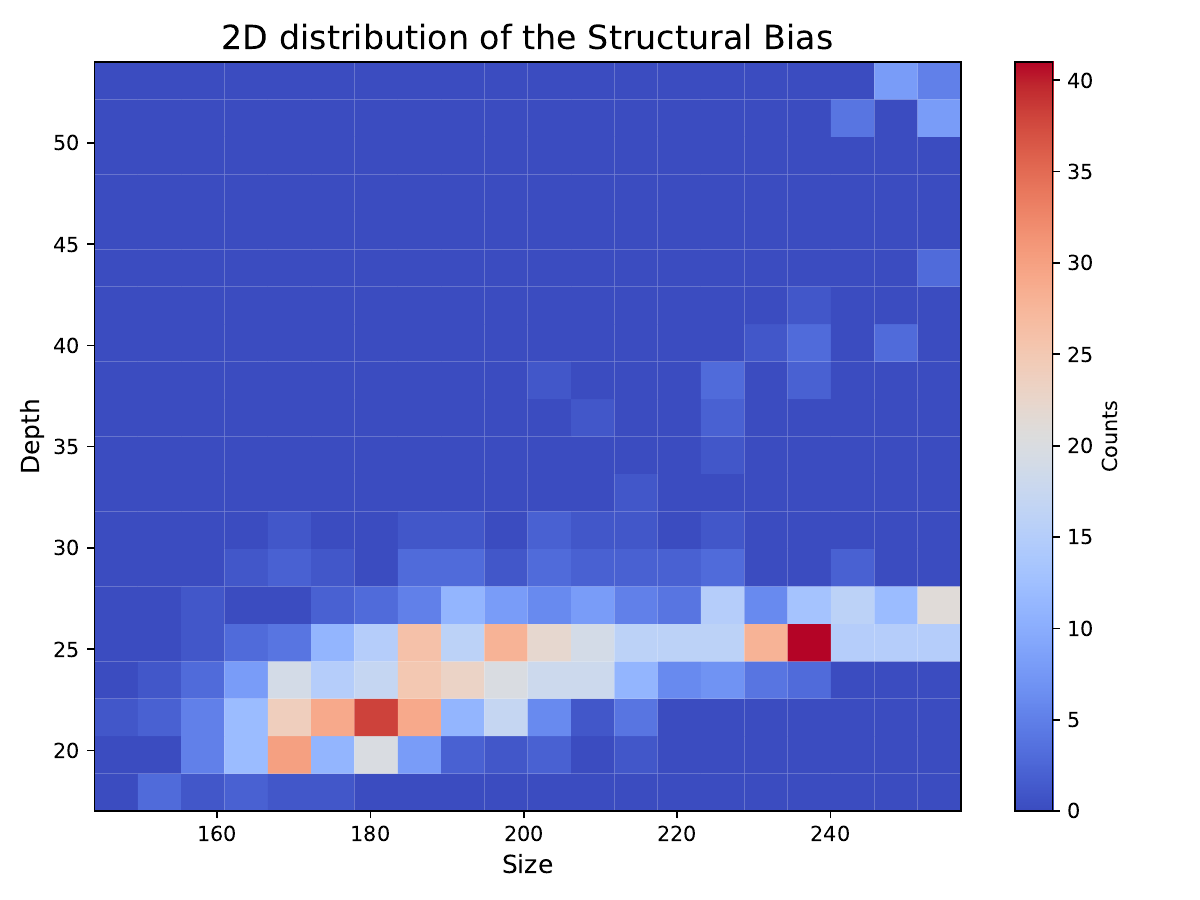}
}
\subfigure[\nlwdel{average degree distr}\nlwadd{degree distr}.]{
\includegraphics[width=0.2\textwidth]{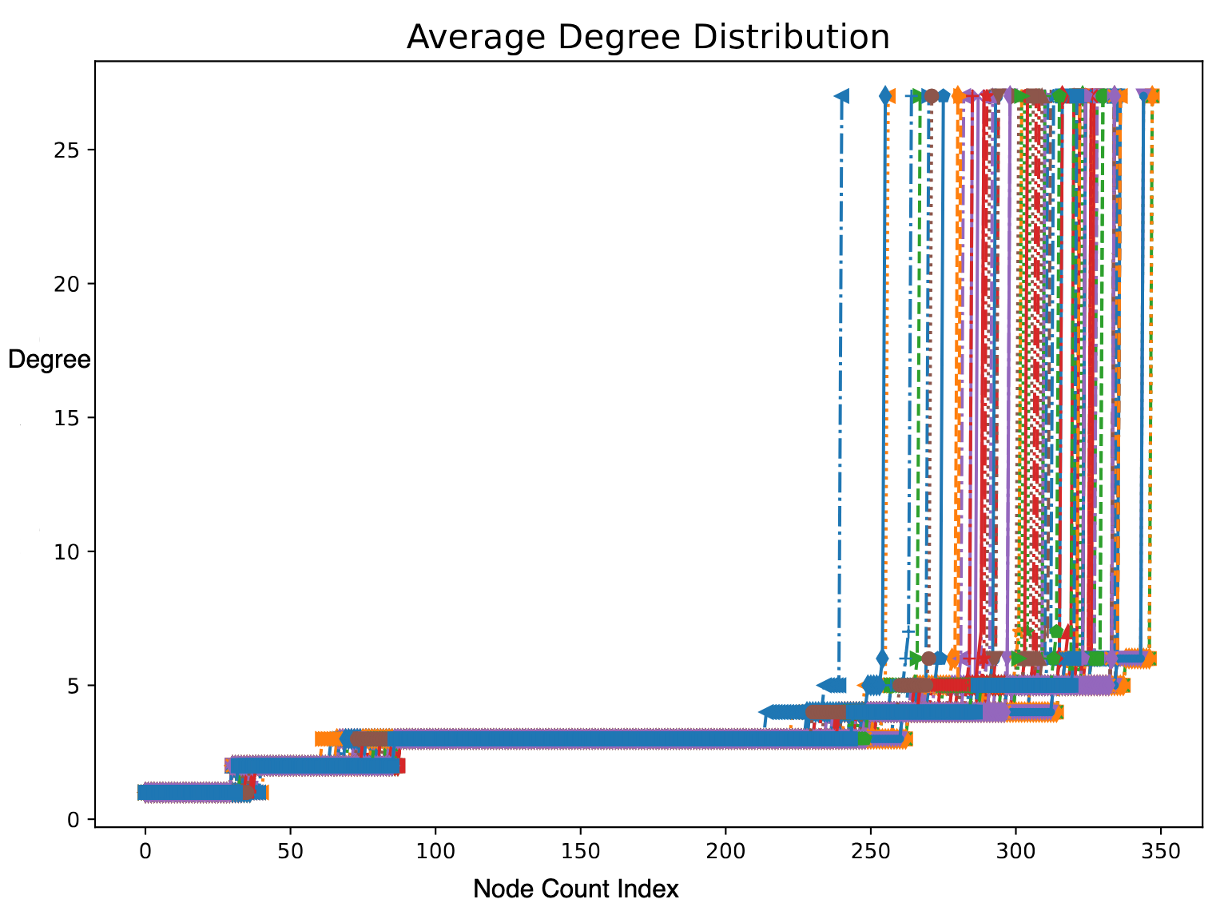}
}
\caption{The structural bias example of \textit{router} design.}
\label{fig:motivation:bias}
\vspace{-1em}
\end{figure}

According to Boolean equivalence theory, distinct graph structures within a Boolean network can yield identical functionality, introducing structural bias~\cite{bias_influence} specific to a given design. 
\cref{fig:motivation:bias} illustrates this structural bias in the design ``router'', demonstrating that variants differing in size, depth, or sorted node degree\footnote{Sorted node degree of a graph can avoid difference by calculating orders} distribution exhibit the same functional behavior. 
Consequently, this equivalence across variants incurs a computational overhead when verifying equivalence to assess structural bias.

\subsubsection{Skeleton}

\begin{figure}[t]
    \centering
    \subfigure[Human Skeleton.]{
    \includegraphics[width=0.2\textwidth]{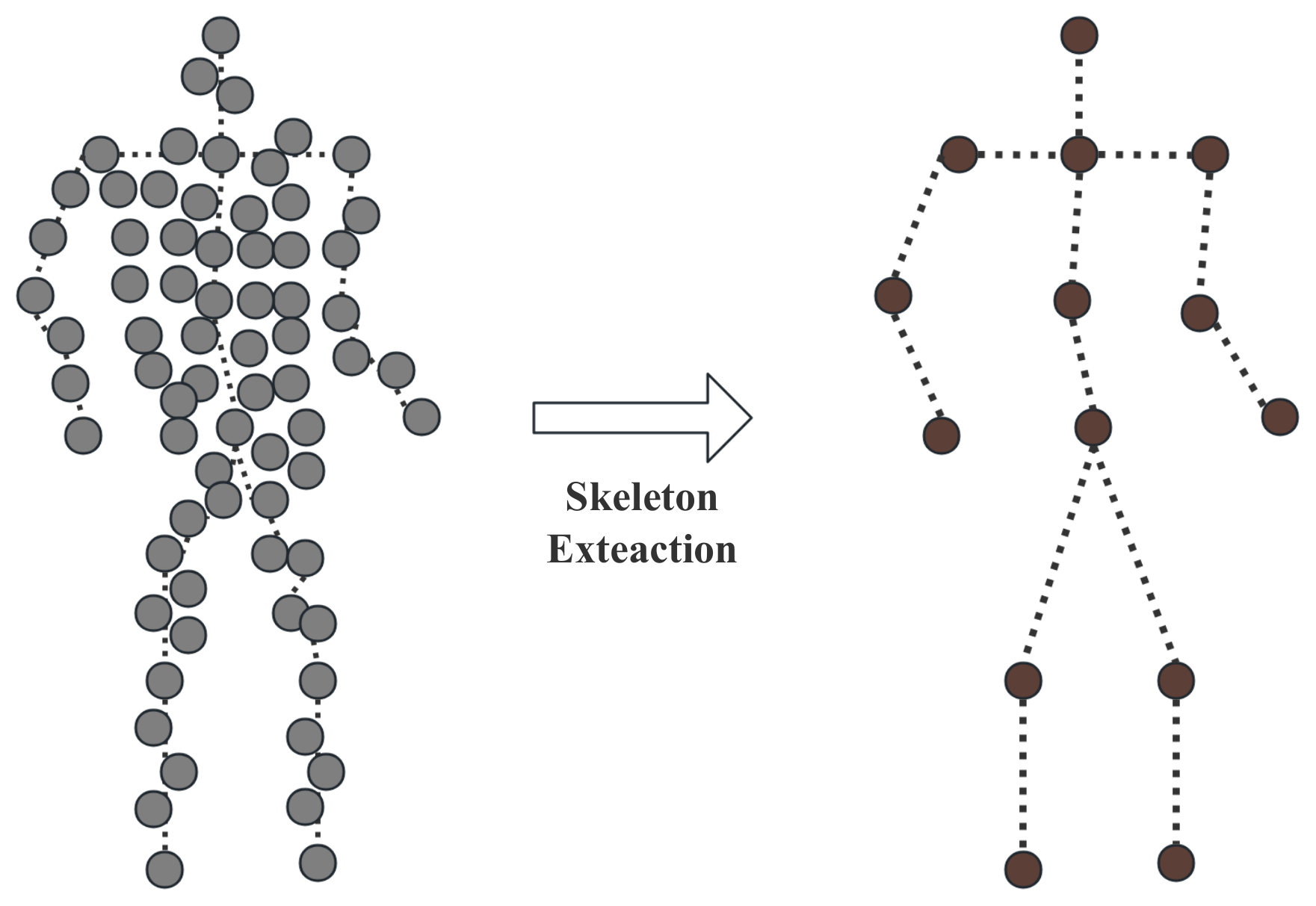}
    }
    \subfigure[Boolean network Skeleton?]{
    \includegraphics[width=0.25\textwidth]{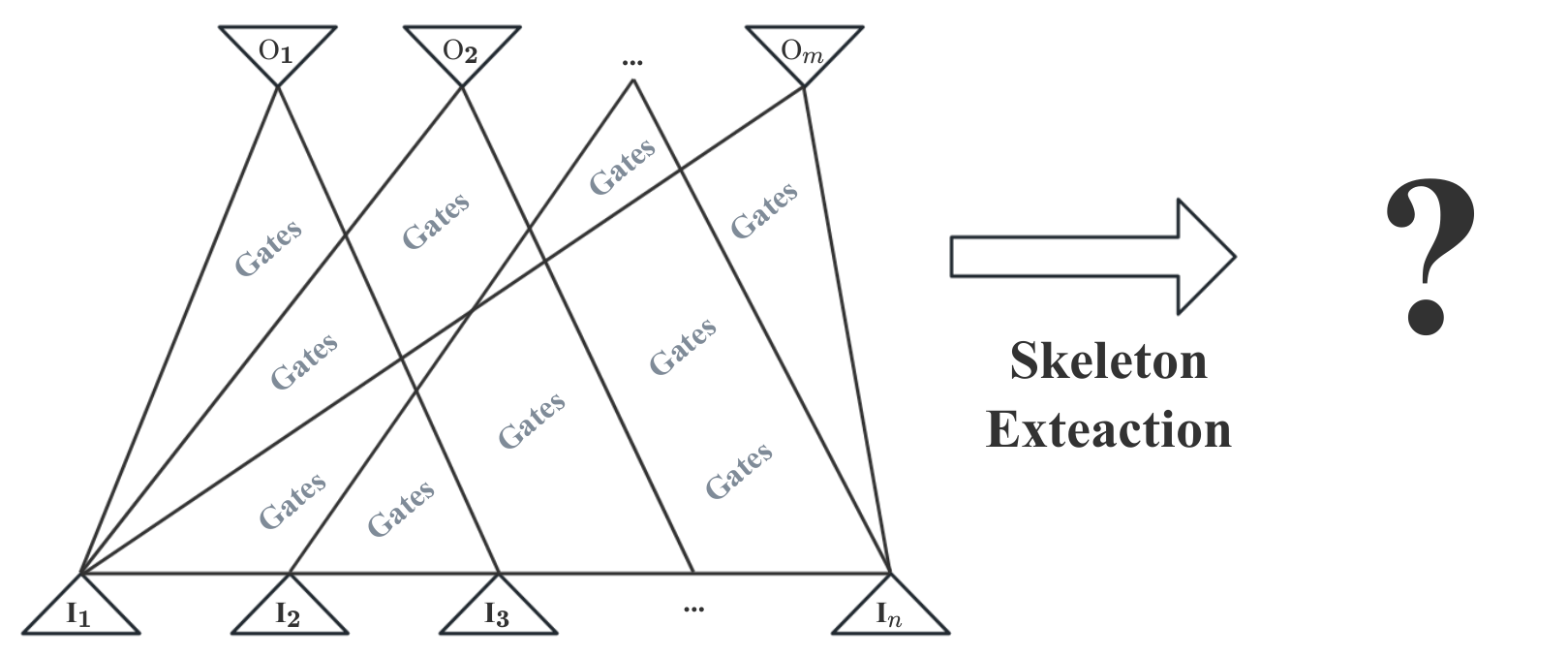}
    }
    \caption{The motivation of the Boolean network skeleton.}
    \label{fig:motivation:skeleton}
\end{figure}


In graph theory, graph reduction~\cite{graph_skeleton, graph_pooling_survery_ijcai23} refers to a technique that reduces the number of nodes or edges in a graph while striving to preserve its fundamental structure and properties. 
This approach can enhance the quality and efficiency of graph-level tasks. 
For instance, \cref{fig:motivation:skeleton} illustrates how the skeleton of a human graph simplifies pose-related tasks by abstracting the essential structure.

In logic synthesis, tasks such as verification and critical path prediction are often hindered by redundant nodes and edges, as well as structural bias. 
Structural bias introduces inconsistencies in graph structure, whereas graph skeletonization aims to abstract the original graph into a more concise representation. 
By integrating the concepts of structural bias and graph skeletonization, it is evident that an abstraction of Boolean networks is essential to improve the performance of related tasks.

\subsection{Related Work}

In graph theory, the skeleton reduction method~\cite{graph_skeleton} offers a robust approach to distill coarse-grained information from a graph while preserving its essential attributes.
Loukas introduced a graph reduction method with spectral and cut guarantees, linking approximation quality to graph properties like degree and eigenvalue distributions~\cite{skeleton_variations}.
Loukas et.al used heavy edge matching for spectral approximations, showing coarse eigenvectors can approximate clustering without refinement~\cite{skeleton_heavy_edge_matching}.
Ron et al. proposed relaxation-based coarsening using algebraic distances for multiscale graph organization~\cite{skeleton_algebraic_distance}.
Livne et al. developed the Lean Algebraic Multigrid (LAMG) method, leveraging node affinity for scalable Laplacian solving~\cite{skeleton_affinity}.
Dörfler et al. introduced Kron reduction for electrical networks, preserving connectivity in reduced graphs~\cite{skeleton_kron_reduction}.
However, Boolean networks differ fundamentally from traditional graphs, presenting a unique challenge: how to adapt skeleton extraction techniques to Boolean networks from the perspective of logic synthesis. 

\section{Problem Formulation and Analysis}
\label{sec:problem}
In this section, we will formally define the Boolean network skeleton problem. 
Then, we analyze the critical factors that substantially impact the efficacy of prospective solutions.

\subsection{Problem Formulation}
\label{sec:problem:formulation}

\begin{definition}[Boolean dependency]
Given a Boolean network $\mathcal{C}$, we say that node $b$ is Boolean dependent on node $a$ if there exists a path from $a$ to $b$ in the DAG and a Boolean function $f$ such that $b = f(a, \dots)$, and the value of $b$ can be determined by the value of $a$\nlwnew{, its negation, or a constant.}.
\label{def:dependency}
\end{definition}
Boolean dependency is based on the concept of reachability in graph theory, a key element for comprehending the flow of information within a DAG.
This dependency underscores functionality, setting it apart from conventional node dependencies in DAGs.
To support graph skeletonization in this work, we introduce the following definition of a Boolean dependency graph tailored for graph operations.

\begin{definition}[Boolean dependency graph]
Boolean dependency graph $\mathcal{G}$ is derived from the Boolean network, where the relationship between nodes refers to Boolean dependency.
\label{def:dependency_graph}
\end{definition}

The Boolean network skeleton problem is based on Boolean dependency graph, and it can be described by:
For any given Boolean network $\mathcal{C}$, the Boolean skeleton problem can be defined as the reduction function \textbf{BNetworkSkeletonize} that maps $\mathcal{C}=(\mathcal{V}, \mathcal{E})$ to its corresponding reduced Boolean dependency graph $\mathcal{G}=(\mathcal{V'}, \mathcal{E'})$: 
\begin{equation}
\begin{aligned}
\mathcal{G} &\gets \textbf{BNetworkSkeletonize}(\mathcal{C}), \\
            & \Rightarrow \mathcal{G} \approx \mathcal{C},
\label{eq:problem_formulation}
\end{aligned}
\end{equation}
where $|\mathcal{V'}| \leq |\mathcal{V}|, |\mathcal{E'}| \leq |\mathcal{E}|$, and the ``$\approx$'' symbol means $\mathcal{G}$ still retains the essential abstraction information of $\mathcal{C}$.

\subsection{Analysis}
\label{sec:problem:analysis}
Following the problem formulation in \cref{eq:problem_formulation}, we address the Boolean network skeleton problem by investigating the following fundamental questions from the skeleton viewpoint:
\begin{enumerate}
    \item What information is fundamental for any given Boolean network?
    \item What information can be directly extracted from a specific Boolean network?
    \item What information can be simplified with minimal impact on the abstraction of a specific Boolean network?
\end{enumerate}

\begin{figure}[t]
    \vspace{-1em}
    \centering
    \begin{minipage}{0.48\textwidth}
        \centering
        {\scriptsize
        \begin{equation}
            \begin{aligned}
            S_0  =& A_0 \oplus B_0,\quad CO_0 = A_0 \wedge B_0 ,\\
            S_i  =& \big{(}A_i \oplus B_i\big{)} \oplus \big{(}CO_{i-1} \wedge (A_{i-1} \oplus B_{i-1})\big{)},~\text{for } i = 1, 2, 3;\\
            CO_i =& \big{(}A_i \wedge B_i\big{)} \vee \big{(}CO_{i-1} \wedge (A_{i-1} \oplus B_{i-1})\big{)} \vee \\                 & \big{(}A_i \wedge CO_{i-1}\big{)} \vee \big{(}B_i \wedge CO_{i-1}\big{)},~\text{for } i = 1, 2, 3.
            \label{eq:adder-4}
            \end{aligned}
        \end{equation}
        \textbf{Boolean algebra of a 4-bit Ripple-Carry Adder}
        }
        
    \end{minipage}%
    \\
    \begin{minipage}{0.48\textwidth}
        \centering
        \includegraphics[width=\textwidth]{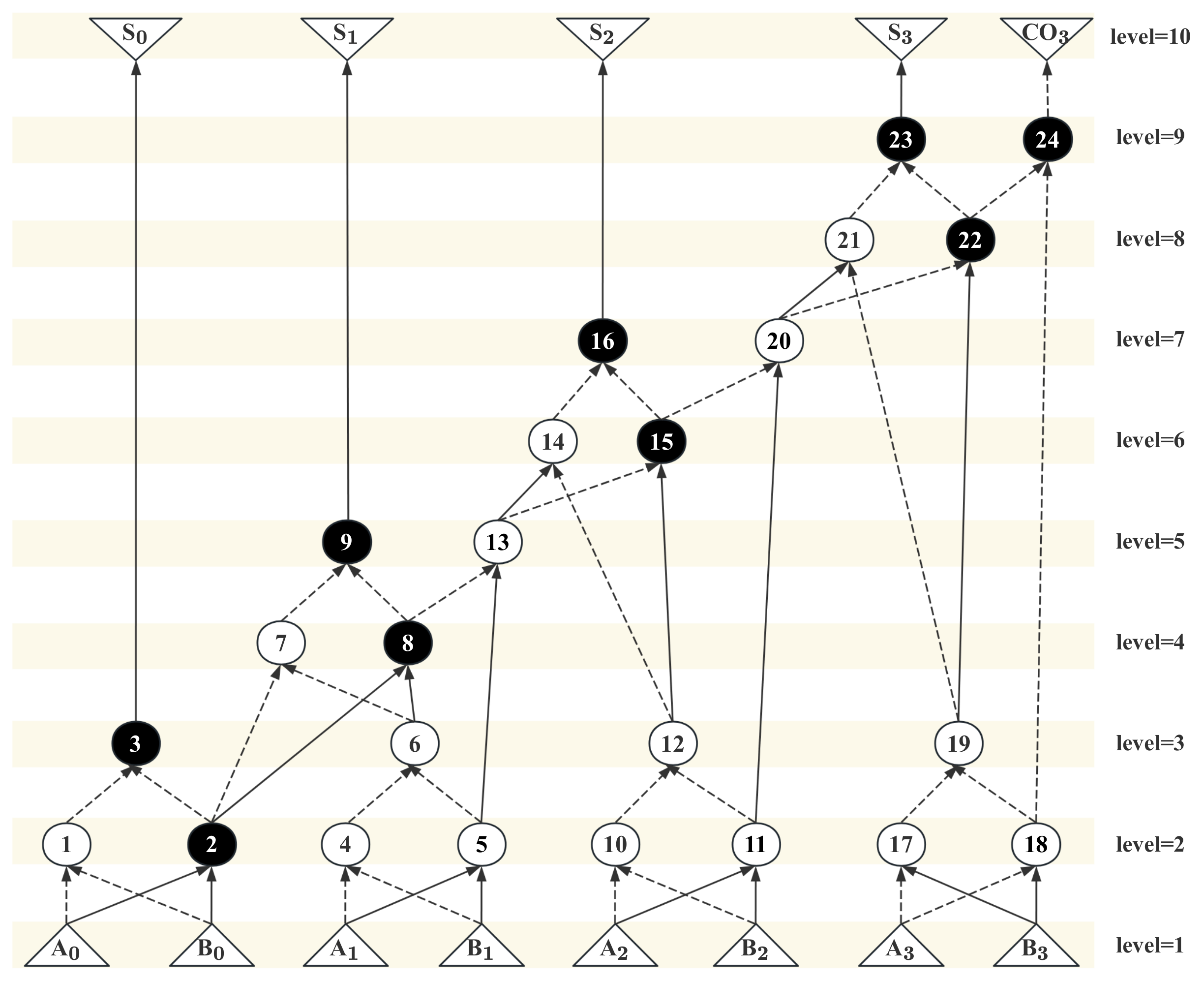}
        \caption{The Boolean network visualization of a 4-bit ripple-carry adder in AIG format.}
        \label{fig:adder-4}
    \end{minipage}
    \vspace{-1em}
\end{figure}

To address the questions outlined above, we present observations derived from a typical 4-bit ripple carry adder (CRA), which subsequently inform a preliminary solution direction. 
For a 4-bit CRA, defined by \(F[4] = A[4] + B[4]\), where \(S\) represents the 4-bit sum vector and \(CO\) denotes the 4-bit vector of carry-out bits, the Boolean function expression is provided in \cref{eq:adder-4}. 
Additionally, \cref{fig:adder-4} illustrates its graph structure to enable deeper analysis. 
To aid this examination, we highlight the sum and carry-out nodes in the graph and assign levels to all nodes. 
Based on this, we derive the following observations and analysis:

\paragraph{Observation 1: Reachability}
For any node \(a\) and its transitive fanin-cone node set \(\mathcal{V}^{\text{TFI}}\), all nodes in \(\mathcal{V}^{\text{TFI}}\) are reachable to \(a\). 
Based on Observation 1, nodes along a single path exhibit varying depths, establishing a Boolean dependency relationship. Conversely, if nodes \(a\) and \(b\) reside on distinct paths, they lack both Boolean dependency and reachability, highlighting the path-specific nature of these relationships.

\paragraph{Observation 2: Boolean Dependency}
Boolean dependency is an inherent attribute of Boolean networks, as established by its definition. 
It implies that the dependency between nodes influences their depth ordering under the unit delay model~\cite{unit_delay_iccad96}. Specifically, if node \(a\) depends on node \(b\), then \(\textit{depth}(a) > \textit{depth}(b)\) must hold, reflecting the directional flow of information in the network.

\paragraph{Observation 3: Reconvergence}
Reconvergence arises as an inevitable consequence of local sharing induced by logic optimization. 
It occurs when signals diverge at a node and subsequently reconverge at a later transitive fanout stage, forming a reconvergence cone. 
This phenomenon generates multiple parallel paths within the cone, complicating the network’s structure and dependency analysis.
From a graph reduction viewpoint, there exist opportunities to merge the reconvergence nodes to simplify Boolean networks.

\paragraph{Observation 4: Dynamic vs. Static Properties}
Boolean equivalence and structural bias underscore the dynamic nature of a Boolean network’s structure for a specific design. Locally, \cref{fig:boolean_equivalence} exemplifies this by depicting distinct graph structures for the function \(F = A \land B \land C \land D\) under two equivalent conditions. 
Globally, logic optimization techniques can iteratively transform an entire Boolean network by substituting local substructures with their Boolean equivalents. 
From both Boolean algebra and graph-theoretic perspectives, the PIs and POs remain fixed for a given design, anchoring their static functionality within dynamic structural variations.

For instance, consider node $n_9$ in \cref{fig:adder-4}, with its transitive fanin-cone node set \(\mathcal{V}^{\text{TFI}}\) = \{$n_{A_0}$, $n_{B_0}$, $n_{A_1}$, $n_{B_1}$, $n_2$, $n_4$, $n_5$, $n_6$, $n_7$, $n_8$\}. 
The depth of each node in \(\mathcal{V}^{\text{TFI}}\) is less than that of node 9, consistent with Observation 1. 
Reconvergence is evident as signals diverge at node 6 and merge at node 9, positioning nodes 7 and 8 on parallel paths. 
Consequently, nodes 7 and 8 exhibit neither Boolean dependency nor reachability, aligning with Observations 2 and 3.
These insights, drawn from the intricate interplay of structure and functionality in Boolean networks, directly address the questions posed earlier. 
The principle of ``less is more'' suggests that a balanced consideration of both graph structure and functionality can yield more reliable skeletonization outcomes. 
This analysis diverges from traditional approaches, which often neglect functionality-related structures, underscoring the critical role of such information in effective skeleton extraction.

\section{\textit{BoolSkeleton}: Boolean Network Skeleton}
\label{sec:method}

\begin{figure*}[t]
    \centering
    \includegraphics[width=1\textwidth]{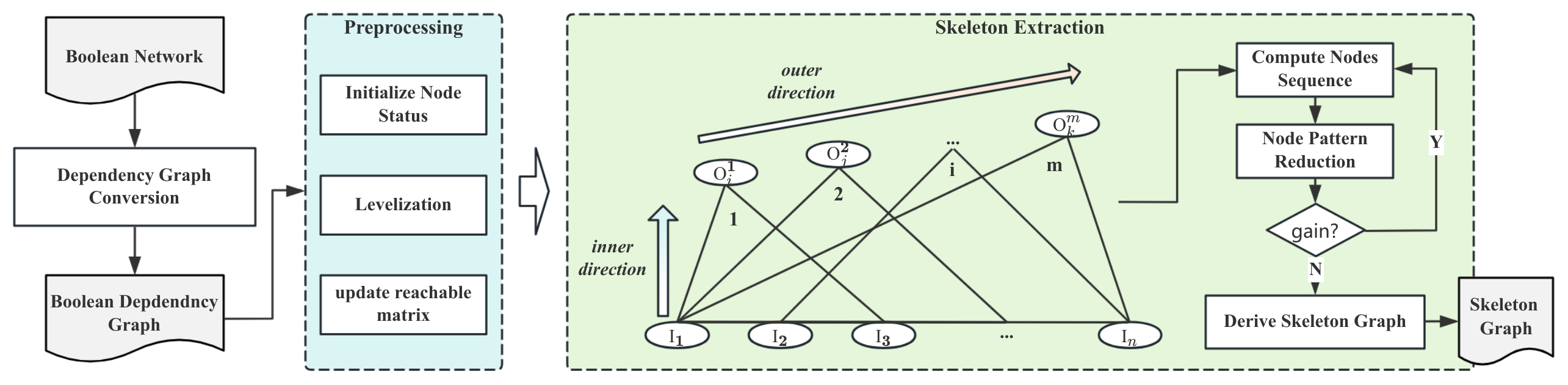}
    \caption{The framework of the proposed Boolean Network Skeleton method, \textit{BoolSkeleton}. }
    \label{fig:framework}
    \vspace{-1em}
\end{figure*}

In this section, we will introduce \textit{BoolSkeleton}, the Boolean network skeleton method as shown in \cref{fig:framework}.
It comprises two primary stages: preprocessing and reduction.
First, the Boolean network is transformed into a Boolean dependency graph with functionality-related node status assignment. 
Then, an iterative node-level pattern reduction approach is employed to get the reduced skeleton.

\subsection{Phase1: Preprocessing}
\label{sec:method:preprocessing}

\subsubsection{Boolean Dependency Graph Recovery}
To facilitate the graph operation, the Boolean network \(\mathcal{C}\) is first transformed into its corresponding Boolean dependency graph \(\mathcal{G}\), a process that systematically maps the network's functional dependencies into a graph representation.
To fully capture the graph structure of \(\mathcal{C}\), we explicitly represent each inverter which is embedded in the edges as a distinct node in \(\mathcal{G}\).
This representation allows the skeletonization algorithm to determine which nodes should be reduced or merged in subsequent steps.
By adopting this approach, all structural details of the Boolean network are preserved while enhancing flexibility for further manipulation.

\subsubsection{Node Status Assignment}
Based on the analysis in \cref{sec:problem:analysis}, PIs and POs constitute critical components of a Boolean network, both from the perspectives of Boolean algebra and graph structure. 
In contrast, internal gates exhibit a dynamic status within the graph structure due to structural bias, rendering them more amenable to reduction or merging during Boolean network skeletonization.

\begin{figure}[!htb]
\vspace{-1em}
    \centering
    \includegraphics[width=0.45\textwidth]{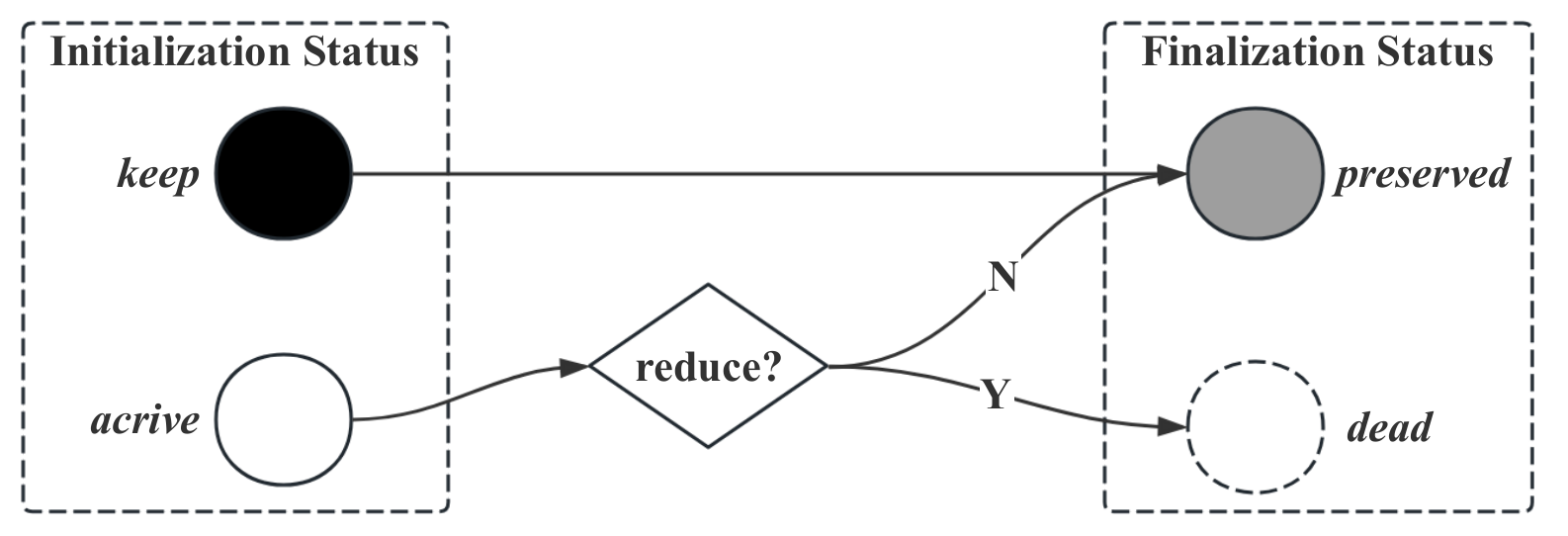}
    \caption{Node status transformation diagram.}
    \label{fig:node-status}
\vspace{-1em}
\end{figure}

To differentiate between static and dynamic information, we define and assign a status to each node, categorized into four types: \{\textit{keep}, \textit{active}, \textit{preserved}, \textit{dead}\}. 
These statuses are illustrated in \cref{fig:node-status} via a transformation diagram. 
Initially, nodes are classified as either \textit{keep} or \textit{active}: \textit{keep} designates nodes with static information that must be retained, while \textit{active} indicates nodes with dynamic information that are candidates for reduction. 
Upon completion of the process, nodes transition to either \textit{preserved} or \textit{dead}: \textit{preserved} signifies nodes retained in the skeleton, and \textit{dead} denotes those eliminated.
During the initialization phase, PIs and POs are assigned the \textit{keep} status, whereas internal gates are designated as \textit{active}:
\begin{equation}
\small
\begin{aligned}
v^s_i& \gets \textit{keep},  &\forall & v_i \in \mathcal{V}^{I};   \\
v^s_o& \gets \textit{keep},  &\forall & v_o \in \mathcal{V}^{O};   \\
v^s_g& \gets \textit{active},&\forall & v_g \in \mathcal{V}^{G},
\end{aligned}
\label{eq:status}
\end{equation}
In subsequent steps, \textit{active} nodes are progressively reassigned to either \textit{preserved} or \textit{dead} based on the reduction process.

\subsection{Phase2: Reduction}
\label{sec:method:reduction}

The reduction process leverages homogeneous pattern reduction applied to the Boolean dependency graph. 
We begin by defining the patterns and their corresponding reduction rules. 
Subsequently, an iterative, fanin-limited, node-level pattern reduction approach is employed to eliminate nodes while preserving the skeleton structure.

\subsubsection{Patterns and the Rule}

\begin{figure}[!tb]
\vspace{-1em}
    \centering
    \subfigure[pattern format.]{
    \includegraphics[scale=0.45]{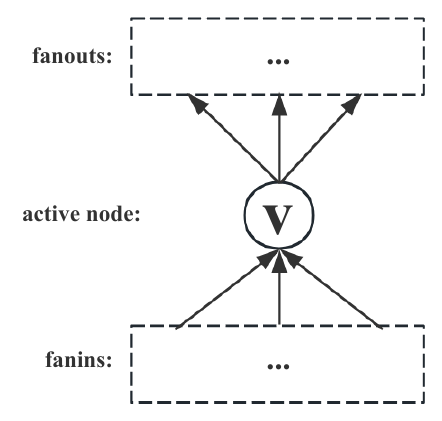}
    }
    \hspace{1cm}
    \subfigure[homo-pattern-1.]{
    \includegraphics[scale=0.45]{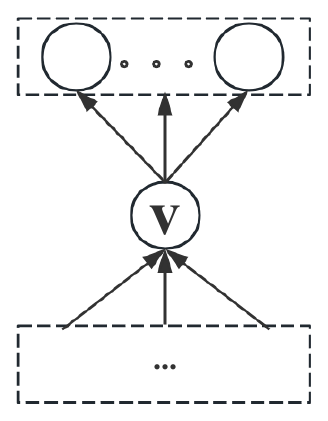}
    }
    \\
    \subfigure[homo-pattern-2.]{
    \includegraphics[scale=0.45]{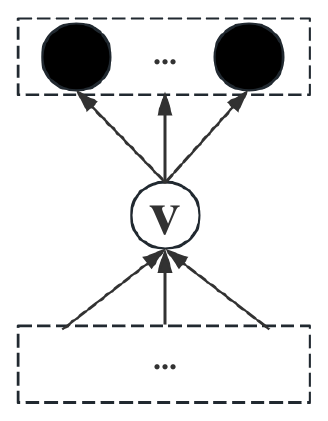}
    }
    \hspace{1cm}
    \subfigure[hetero-pattern]{
    \includegraphics[scale=0.45]{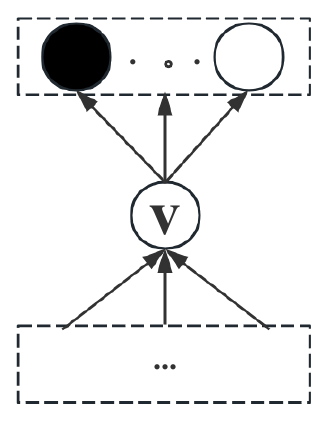}
    }
    \caption{The illustration of the defined patterns.}
    \label{fig:patterns}
\vspace{-1em}
\end{figure}

\begin{definition}[Pattern]
\label{def:pattern}
A pattern \(\mathbf{P} = (v, \mathcal{V}^{\text{fanin}}, \mathcal{V}^{\text{fanout}})\) is defined as a node-level DAG subgraph comprising three components: a central node \(v\), its in-degree nodes \(\mathcal{V}^{\text{fanin}}\), and its out-degree nodes \(\mathcal{V}^{\text{fanout}}\).
\end{definition}

\cref{def:pattern} establishes the node-level pattern structure, with its graphical representation depicted in \cref{fig:patterns}~(a). 
The proposed method targets node-level reduction to derive an abstraction of the original graph, and accordingly, the defined pattern format emphasizes a node-centric perspective. 
In practice, patterns are classified into two categories based on the consistency of the statuses of \(\mathcal{V}^{\text{fanout}}\) in \(\mathbf{P}\): heterogeneous and homogeneous patterns, as detailed below:

\begin{definition}[Heterogeneous Pattern]
\label{def:pattern:hetero}
A pattern \(\mathbf{P}\) is termed a heterogeneous pattern, denoted \(\mathbf{P}_{\text{hetero}}\), if its fanout nodes \(\mathcal{V}^{\text{fanout}}\) exhibit distinct statuses.
\end{definition}

\begin{definition}[Homogeneous Pattern]
\label{def:pattern:homo}
A pattern \(\mathbf{P}\) is termed a homogeneous pattern, denoted \(\mathbf{P}_{\text{homo}}\), if its fanout nodes \(\mathcal{V}^{\text{fanout}}\) share the same status.
\end{definition}

\begin{figure}[t]
\vspace{-1em}
    \centering
    \includegraphics[width=0.5\textwidth]{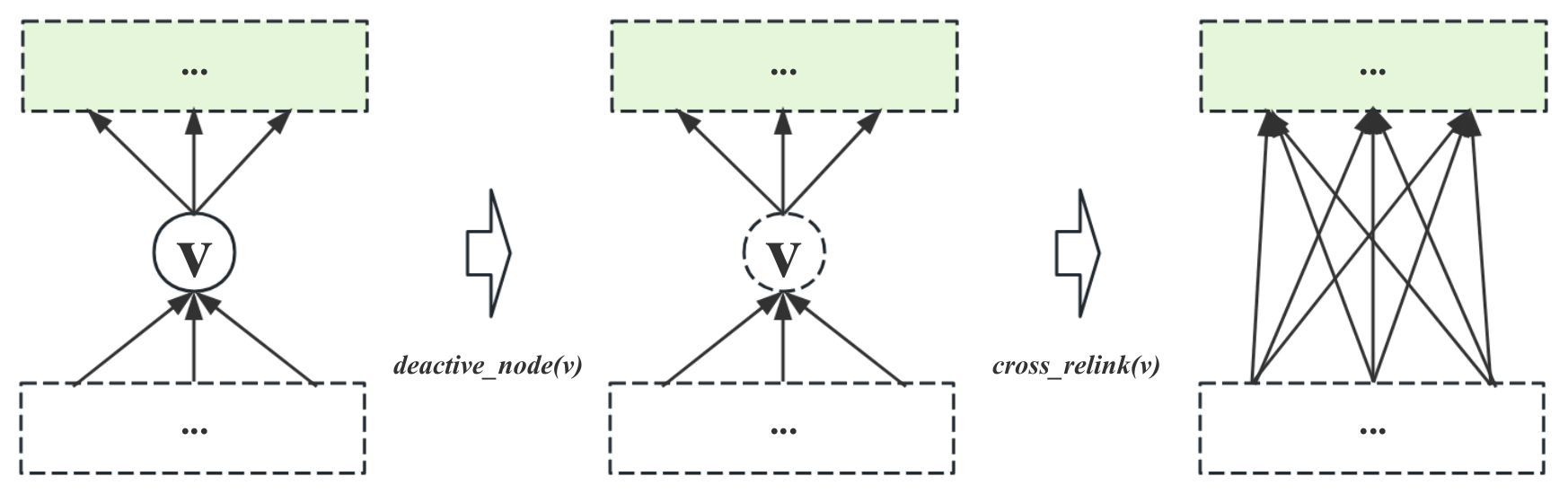}
    \caption{The pattern reduction rule.}
    \label{fig:reduction_rule}
\vspace{-1em}
\end{figure}

\begin{table}[!tbh]
\vspace{-1em}
\centering
\small
\caption{The occurrence probability of the defined patterns.}
\begin{tabular}{ c | c c c }
\toprule
\diagbox{\textbf{Order}}{\textbf{Pattern}} & $\mathbf{P}_{homo}^{1}$   & $\mathbf{P}_{homo}^{2}$ & $\mathbf{P}_{hetero}$   \\
\midrule
Topo-order~(PI $\rightarrow$ PO)                   &    high        &  low            & low              \\
Reverse Topo-order~(PO $\rightarrow$ PI)           &    low         &  high           & high             \\
\bottomrule
\end{tabular}
\label{tab:pattern_occurrence}
\vspace{-1em}
\end{table}

\begin{definition}[Reduction rule]
\label{def:reduction_rule}
The reduction rule $\sigma$ is defined as a mapping function that maps the pattern to its target subgraph $\mathbf{P'}$:
\begin{equation}
\small
\begin{aligned}
\mathbf{P'} & \gets \sigma(\mathbf{P}),
\end{aligned}
\end{equation}
where \nlwdel{$\mathbf{P} = v \cup \mathcal{V}^{in} \cup \mathcal{V}^{out}$,} $\mathbf{P'} = \mathcal{V}^{in} \cup \mathcal{V}^{out}$\nlwdel{, $\mathbf{P}^{'} \subset \mathbf{P}$ and $|\mathbf{P}^{'}| = |\mathbf{P}| - 1$}.
\end{definition}

\cref{fig:patterns}(b) and \cref{fig:patterns}(c) show the homogeneous patterns $P^1_{homo}$ and $P^2_{homo}$, respectively, and \cref{fig:patterns}(d) shows the heterogeneous pattern $P_{hetero}$.
\cref{tab:pattern_occurrence} demonstrates the probability of the occurrence of the defined patterns according to the assigned node status.
It should be noted that not all the patterns are used for the reduction.
The heterogeneous pattern can distinguish the \nlwdel{Boolean dependency}\nlwadd{depth order} between different nodes, while the homogeneous pattern can not; it mainly propagates the signals to the nodes with the same status.
These properties are discussed in Property. \ref{property:pattern:homo} and Property. \ref{property:pattern:hetero}.
It should be noted that the pattern reduction rule $\sigma$ is focused on node-level reduction as shown in \cref{fig:reduction_rule}.
The \textbf{PatternReduction} is defined as the node reduction operator on an ``active'' node $v$ of the Boolean dependency graph $\mathcal{G}$ by the reduction rule $\sigma$: 
\begin{equation}
\small
\textbf{PatternReduction}\Big(\mathcal{G}, v, \mathbi{A}, \mathbi{R}, \mathbf{P}, \sigma(\mathbf{P}), K\Big),
\end{equation}
where $K$ is the limitation of the fanin size of a pattern, which can be used to control the graph coarsening ratio.

For any given Boolean dependency graph $\mathcal{G}$ with the adjacent matrix $\mathbf{A}$ and the reachable matrix $R$, current processed node $v_k$, the pattern reduction can be formulated:
\begin{equation}
\small
\begin{aligned}
v_k^{s} & \gets dead, \\
\forall v_i \in \mathcal{V}_{v_k}^{fanin}, \forall v_o \in \mathcal{V}_{v_k}^{fanout}. &
\begin{cases}
1.~\mathbi{A}_{v_i,v_k} \gets 0, \\ 
2.~\mathbi{A}_{v_k,v_o} \gets 0, \\ 
3.~\mathbi{A}_{v_i,v_o} \gets 1~(\mathbi{A}_{v_i,v_o}=0), \\ 
4.~\mathbi{R}_{v_i,v_o} \gets 1, \\
\end{cases}
\end{aligned}
\label{eq:reduction_rule}
\end{equation}
As illustrated by \cref{fig:reduction_rule} and \cref{eq:reduction_rule}, the process begins by updating the status of node \(v_k\) to \textit{dead} and severing all its connections.
Subsequently, for each pair of nodes from \(\mathcal{V}_{v_k}^{\text{fanin}}\) to \(\mathcal{V}_{v_k}^{\text{fanout}}\), a direct edge is established (\(\mathbf{A}_{v_i, v_o} = 1\)) if no prior path existed—where a path implies reachability—and the reachability matrix is updated accordingly (\(\mathbf{R}_{v_i, v_o} = 1\)).

\subsubsection{Iterative Node-level Pattern Reduction Approach}
\label{sec:algorithm}

The Boolean dependency graph reduction can be addressed by the iteratively fanin-limited node-level homogeneous pattern reduction:
\begin{equation}
\small
\label{eq:circuit_pooling:2}
\begin{aligned}
\mathcal{G} & \gets \textbf{BNetworkSkeletonize}(\mathcal{C}), \\
            & \Rightarrow \mathcal{G}^{'} \gets \underset{i=0}{\stackrel{n-1}{\textbf{LimitedPatternReduction}}}\big{(}\mathcal{G}, v_i, \mathbi{A}, \mathbi{R}, \mathbf{P}_{homo}, \sigma, K\big{)}
\end{aligned}
\end{equation}

\begin{algorithm}[t]
\caption{Iterative Node-level Pattern Reduction}
\label{alg:skeleton}
\begin{algorithmic}[1]
    \REQUIRE Boolean dependency graph $\mathcal{G}$, homogeneous pattern $\mathbf{P}_{homo}$, reduction rule $\sigma$, limitation $K$
    \ENSURE  Skeleton graph $\mathcal{G}^{'}$
    \STATE $\mathbi{A}, \mathbi{R} \gets \text{update\_graph\_matrix}(G)$
    \STATE levelization($\mathcal{G}$)
    \STATE initialize\_node\_status($\mathcal{G}$)
    \STATE $\mathcal{V}^{PO'} \gets$ sort\_nodes\_by\_level\_ascending($\mathcal{V}^{PO}$)
    \WHILE{\textit{true}}
        \STATE set $count \gets 0$        
        \FOR{$v_o$ in $\mathcal{V}^{PO'}$}
            \STATE $\mathcal{V}^\textit{fanincone} \gets \text{collect\_fanincone\_by\_dfs}(\mathcal{G}, v_o)$
            \FOR{$v$ in $\mathcal{V}^\textit{fanincone}$}
                \IF{get\_fanin\_size($\mathcal{G}, v$) $\geq$ $K$}
                    \STATE set\_node\_status\_preserved($\mathcal{G}, v$)
                    \STATE continue
                \ENDIF
                \STATE $count~\text{+=}~\textbf{PatternReduction}(\mathcal{G}, v, \mathbi{A}, \mathbi{R}, \mathbf{P}_{homo}, \sigma)$
            \ENDFOR
        \ENDFOR
        \IF{$count$ = 0}
            \STATE break       
        \ENDIF
    \ENDWHILE
    \STATE $\mathcal{G}^{'} \gets \text{skeleton\_nodes\_collection}(\mathcal{G})$ 
    \STATE {\bfseries Return:} $\mathcal{G}^{'}$
\end{algorithmic}
\end{algorithm}

\cref{eq:circuit_pooling:2} provides a detailed formulation of the Boolean network skeleton solution through an iterative, node-level, fanin-limited pattern reduction approach.
\cref{alg:skeleton}  takes inputs of the converted Boolean dependency graph \(\mathcal{G}\), the homogeneous pattern \(\mathbf{P}_{\text{homo}}\), the reduction rule \(\sigma\), and a fanin size limit \(K\), producing the extracted skeleton graph \(\mathcal{G}'\) as output. 
The process begins with a preprocessing phase (lines 1–3), which involves computing the adjacency matrix \(\mathbf{A}\) and reachability matrix \(\mathbf{R}\), updating each node’s level using the unit delay model~\cite{unit_delay_iccad96}, and initializing node statuses as \cref{eq:status}.
Subsequent steps focus on node-level reduction. Initially, primary output nodes are collected and sorted into \(\mathcal{V}^{\text{PO}'}\) in ascending order of depth. 
The unit delay model \nlwdel{ensures that processing nodes in depth order preserves Boolean dependency}\nlwadd{guarantees that the order of processing nodes can maintain depth order}. 
Within the while loop, a counter \(count\) tracks the number of nodes reassigned to \textit{dead} in each iteration. 
The fanin-limited pattern reduction is applied in topological order to the transitive fanin-cone \(\mathcal{V}^{\text{fanincone}}\) of a specific primary output node. 
If a node’s fanin size exceeds \(K\), its status is set to \textit{preserved} (lines 10–13); otherwise, it is evaluated for reduction via the \(\textit{try\_node\_pattern\_reduction}()\) function, detailed in \cref{alg:node_reduction} (line 15). 
The loop terminates when no nodes are reduced in an iteration (lines 17–19). 
Finally, the skeleton graph \(\mathcal{G}'\) is constructed by removing all \textit{dead} nodes and connecting the remaining \textit{preserved} nodes of \(\mathcal{G}\) (line 21).

\begin{algorithm}[t]
\caption{Try Homogeneous Pattern Reduction}
\label{alg:node_reduction}
\begin{algorithmic}[1]
    \REQUIRE Boolean Dependency Graph $\mathcal{G}$, adjacent matrix $\mathbi{A}$, reachable matrix $\mathbi{R}$, node $v$, pattern $\mathbf{P}$, reduction rule $\sigma$
    \ENSURE The reassigned dead node count: $count$
    \STATE Initialize $count \gets 0$
    \IF{not is\_node\_status\_activate($\mathcal{G}$, $v$)}
        \STATE {\bfseries Return:} $count$
    \ENDIF
    \IF{is\_match\_pattern($\mathcal{G}, v, \mathbf{P}$)}
        \STATE set\_node\_status\_dead($\mathcal{G}, v$)
        \STATE reset\_adjacent\_matrix\_at($v, \mathbi{A}$)
        \STATE reset\_reachable\_matrix\_at($v, \mathbi{R}$)
        \STATE add\_cross\_edges\_and\_reachability($v, \mathbi{A}, \mathbi{R}$)
        \STATE $count \gets 1$
    \ELSE
        \STATE set\_node\_status\_preserved($\mathcal{G}, v$)
    \ENDIF
    \STATE {\bfseries Return:} $count$
\end{algorithmic}
\end{algorithm}

\cref{alg:node_reduction} elaborates the homogeneous pattern reduction process, adhering to \cref{eq:reduction_rule}, with its application context illustrated in \cref{alg:skeleton} (line 14). 
The variable \(count\) indicates whether the current node is reassigned to \textit{dead}, returning 1 for true and 0 for false. 
The procedure first verifies that the node \(v\)’s status is \textit{active}; if not, it exits immediately. 
If node \(v\) and its surrounding subgraph match the specified pattern, a node-level reduction is executed following the four steps in \cref{eq:reduction_rule} (lines 6–9). 
Notably, an edge between \(v\)’s fanin and fanout nodes is added only if \nlwdel{no prior reachability exists}\nlwadd{none of the other fanins of the fanout are reachable by the node $v$}, preventing the introduction of redundant or nested edges that could compromise the skeleton graph’s integrity.


\subsubsection{Theorem Analysis}
Based on the definitions of the proposed patterns and their reduction rule, we derive several key properties that echo the Boolean network analysis in \cref{sec:problem:analysis}.
\nlwnew{We assume that the circuit does not contain the disjoint support outputs.}

\begin{property}
\label{property:pattern:homo}
The homogeneous pattern \(\mathbf{P}_{\text{homo}}\) does not alter the depth order among primary output (PO) nodes. As depicted in \cref{fig:patterns}~(b) and (c), all fanout nodes share the same status. Per Property~\ref{property:pattern:hetero}, the central node \(v_c\) does not contribute to distinguishing its fanout cones, thus preserving the relative depth order of POs.
\end{property}

\begin{property}
\label{property:pattern:hetero}
The heterogeneous pattern \(\mathbf{P}_{\text{hetero}}\) preserves the depth order among PO nodes. As illustrated in \cref{fig:patterns}~(d), a heterogeneous pattern \(\mathbf{P}_{\text{hetero}}\) includes at least one \textit{keep} node \(v_k \in \mathcal{V}^{\text{fanout}}\) and one \textit{active} node \(v_a \in \mathcal{V}^{\text{fanout}}\). Assuming \(depth(v_c) = d_c\) for the central node \(v_c\), and excluding the influence of other nodes on this pattern’s depth, we have \(depth(v_k) = d_c + 1\) and \(depth(v_a) = d_c + 1\). If \(v_a\) targets a PO node \(v_{\text{po}}\), then,according to unit delay model in \cref{eq:unit_delay_model}, \(depth(v_{\text{po}}) = \max_{(v_i \to v_{\text{po}}) \in \mathcal{E}} (depth(v_i) + 1)\). Thus, \(depth(v_{\text{po}}) \geq depth(v_a) + 1 = d_c + 2 > depth(v_k)\), ensuring the depth order is maintained.
\end{property}

\begin{figure*}[t]
    \centering
    \begin{minipage}{0.55\textwidth}
        \centering
        \includegraphics[width=0.6\textwidth]{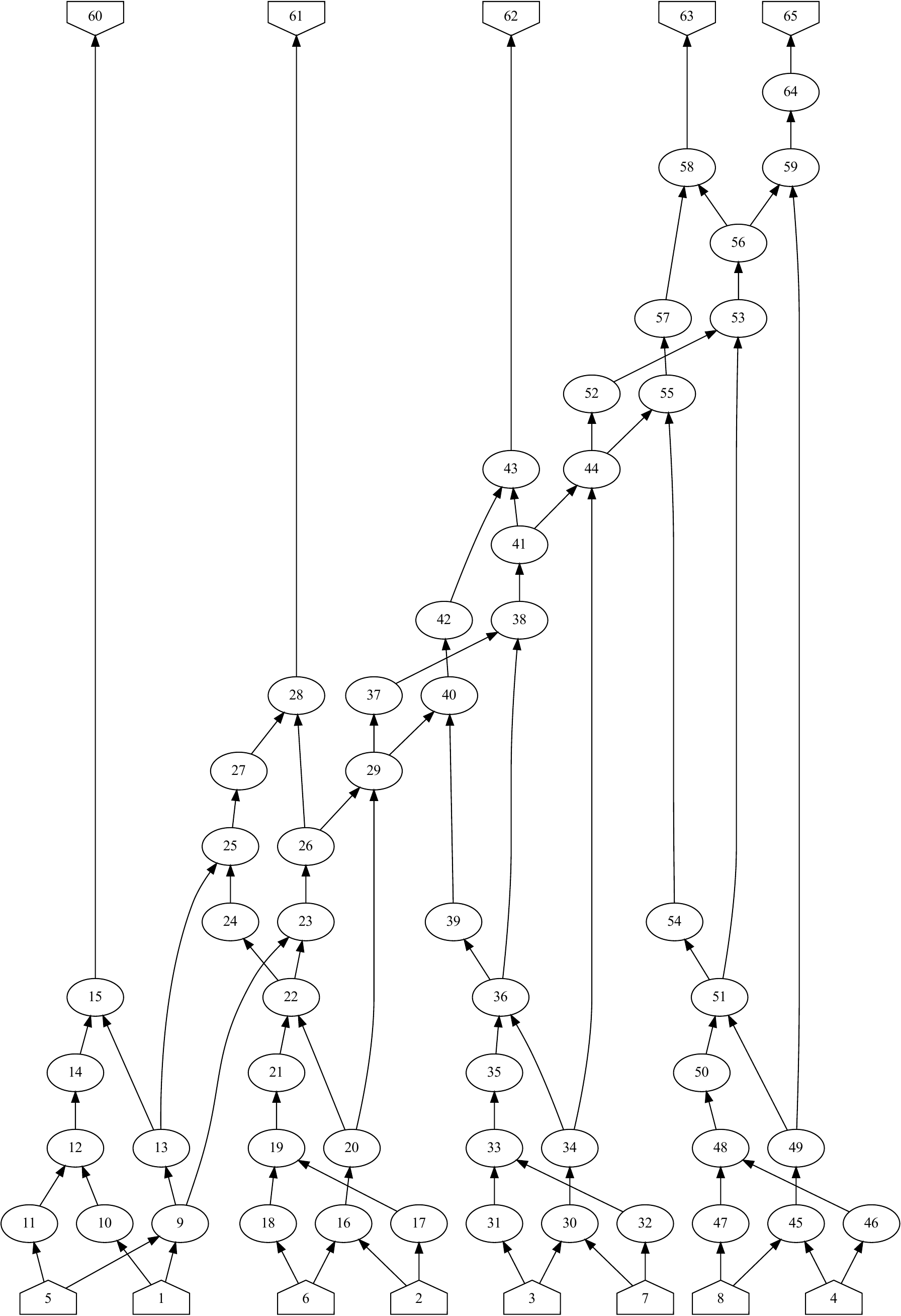}
        \caption*{(a) Skeleton~(K=1)~\footnote{K=1 means the Boolean network is first translated into a Boolean dependency graph.}}
        \includegraphics[width=0.6\textwidth]{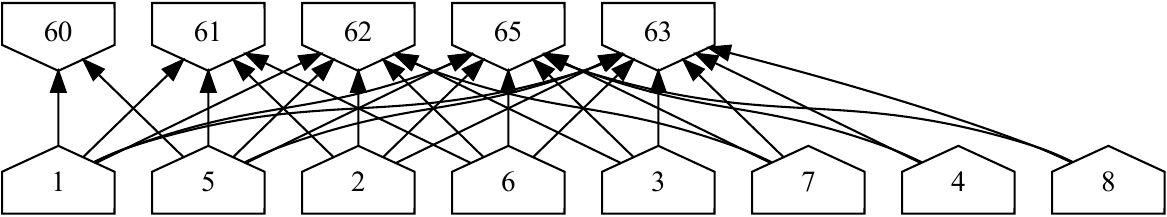}
        \caption*{(e) Skeleton~(K=$\infty$~(unlimited))}
    \end{minipage}
    \hfill
    \begin{minipage}{0.44\textwidth}
        \centering
        \includegraphics[width=0.6\textwidth]{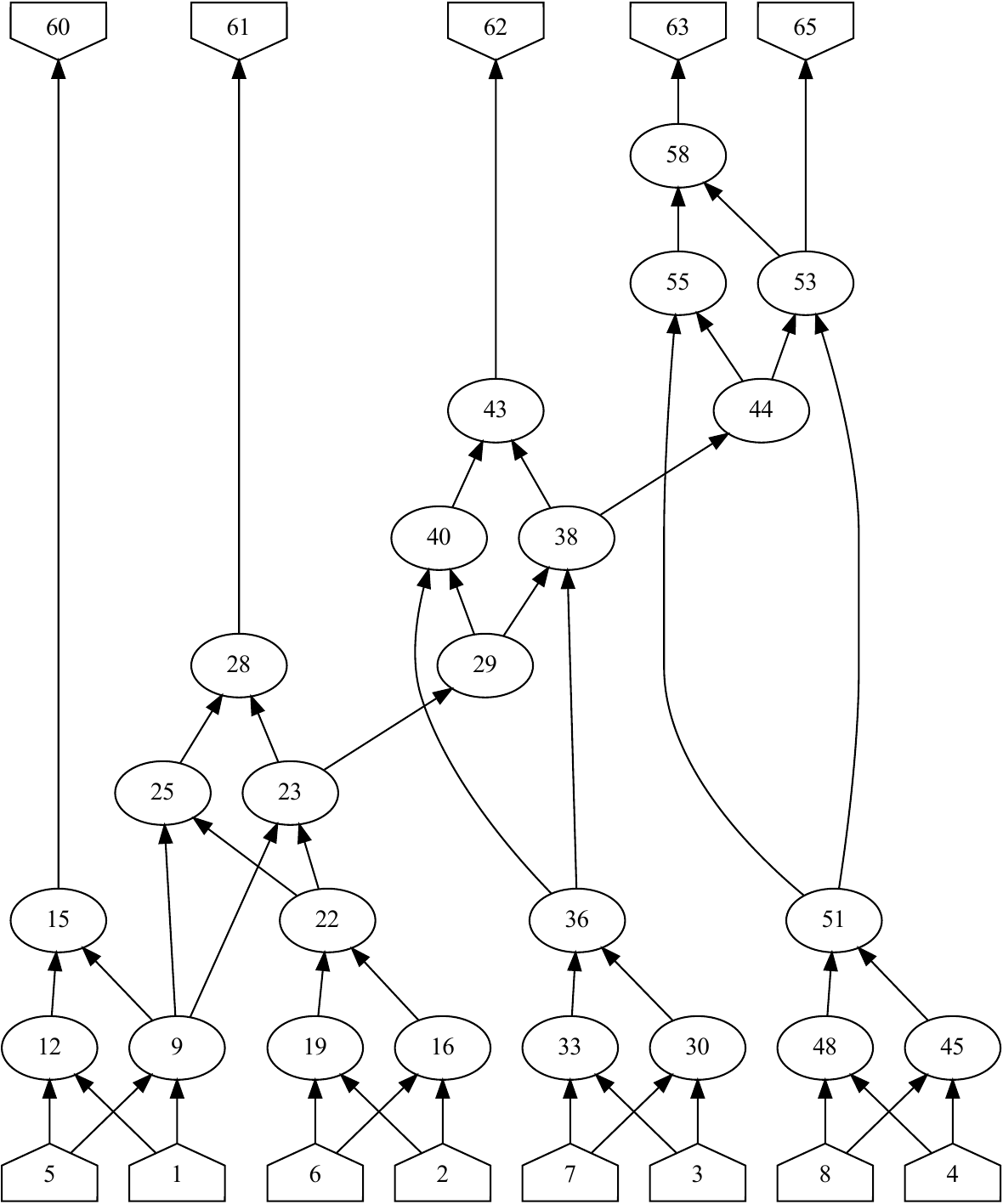}
        \caption*{(b) Skeleton~(K=2)}
        \includegraphics[width=0.6\textwidth]{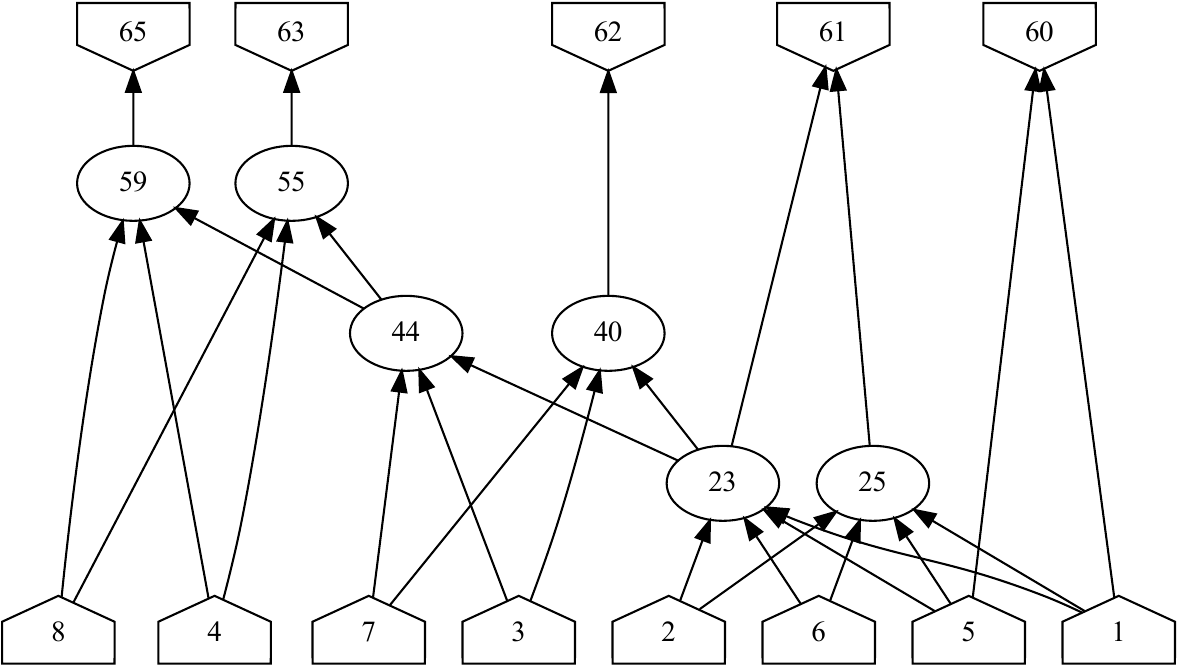}
        \caption*{(c) Skeleton~(K=3)}
        \includegraphics[width=0.6\textwidth]{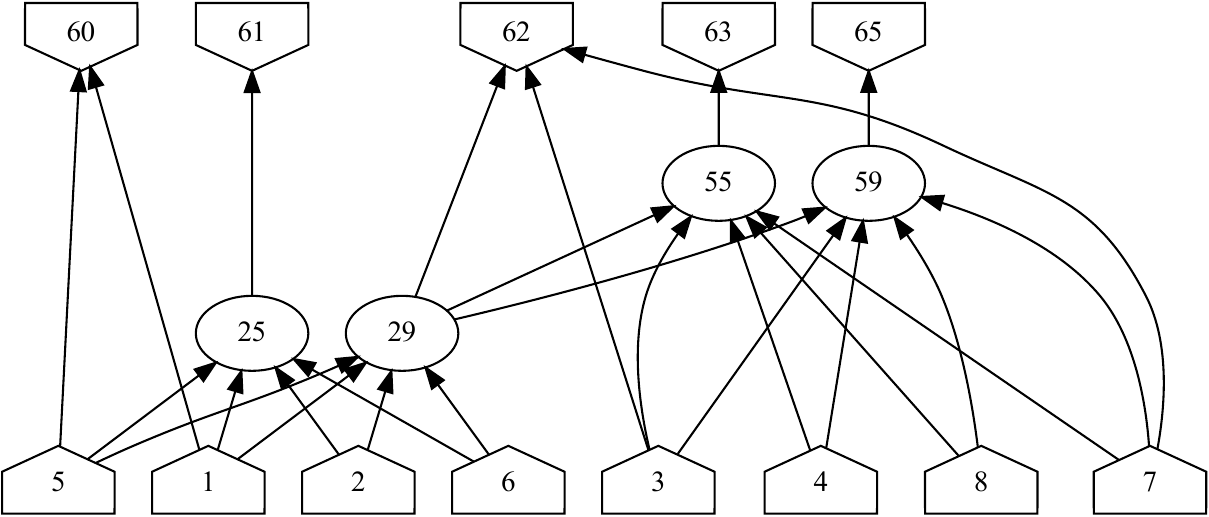}
        \caption*{(d) Skeleton~(K=4)}
    \end{minipage}
    \caption{The case study of the skeleton with different fanin limitation of the 4-bit ripple carry adder in \cref{fig:adder-4}.}
    \label{fig:case}
    \vspace{-1em}
\end{figure*}

\begin{proposition}
\label{prop:reachability}
The \(\textbf{PatternReduction}\) operator preserves both the reachability relations and topological order of the remaining nodes in the Boolean dependency graph \(\mathcal{G}\). Additionally, this preserved reachability ensures the original fanout propagation across the transitive fanin (TFI) and fanout (TFO) cones of the reduced node.
\end{proposition}

\begin{proof}[Proof of Proposition~\ref{prop:reachability}]
Consider a node \(v\), for topological order, the transitive property of partial ordering states that if \(v_i \preceq v\) and \(v \preceq v_o\), then \(v_i \preceq v_o\). This is verified as follows:
\begin{equation}
\small
\nonumber
\begin{aligned}
&\forall v_i \in \mathcal{V}_{v}^{\text{fanin}}, \, v_i \preceq v, \\
&\forall v_o \in \mathcal{V}_{v}^{\text{fanout}}, \, v \preceq v_o, \\
&\Rightarrow \forall v_i \in \mathcal{V}_{v}^{\text{fanin}}, \, \forall v_o \in \mathcal{V}_{v}^{\text{fanout}}, \, v_i \preceq v_o.
\end{aligned}
\end{equation}
Thus, iterative applications of \(\textbf{PatternReduction}\) preserve both reachability and topological order for all retained nodes.
\end{proof}

We conclude that the \(\textbf{PatternReduction}\) operator is both reachability- and topology-aware, making it suitable for applications requiring these properties.

\subsection{Case Study}
\label{sec:method:case}



\cref{fig:case} presents a case study of the Boolean network skeleton applied to a 4-bit ripple carry adder design, as shown in \cref{fig:adder-4}, under different coarsening ratios by the fanin size limitation $K$ of the pattern. 
Specifically, \cref{fig:case}~(a) depicts the Boolean dependency graph directly derived from the initial Boolean network~(AIG).
This case study yields two primary observations:
\begin{enumerate}
    \item As the fanin constraint \(K\) increases, the graph coarsening ratio rises, reducing the number of nodes and yielding a more pronounced skeleton structure.
    \item The similarity between the skeleton and the original Boolean dependency graph decreases, as the retained information diminishes with a greater coarsening ratio.
\end{enumerate}

These observations highlight a key insight: the fanin constraint \(K\) in \textit{BoolSkeleton} acts as a control parameter for the graph coarsening ratio. 
The optimal coarsening level varies across tasks; for instance, functionality-related tasks prioritize preserving heterogeneous structures, while Boolean network profiling tasks emphasize balancing local information with the global skeleton. 
To assess the divergence between the skeleton and the original Boolean dependency graph as \(K\) increases, graph similarity metrics, such as those based on spectral analysis~\cite{graphsimilarity_cc2017_evaluation, graphsimilarity_kdd2021_evaluation}, provide a robust evaluation tool.
Selecting an appropriate \(K\) thus depends on task-specific requirements, and while a universal value remains elusive, such metrics enable tailored similarity assessments. 
Further exploration of this variability is deferred to subsequent task-specific analyses.


\section{Empirical Evaluation}
\label{sec:task}

In this section, we will show the effectiveness of \textit{BoolSkeleton} by conducting two validation evaluations: compression and classification,  and two downstream tasks: critical path analysis and timing prediction.

\subsection{Setup}
\label{sec:experiment:setup}

\paragraph{Environment.}
The codes of \textit{BoolSkeleton} are written in C++, and the following tasks are conducted in Python.
The experiments were conducted using the following hardware and software configuration:
\textit{(Hardware)} Intel Xeon Platinum 8380 CPU (160 cores), 512 GB RAM, NVIDIA A100 GPU (40 GB VRAM); 
\textit{(Software)} Ubuntu 20.04.6, Python 3.8, PyTorch 2.0.1, CUDA 12.0, torch\_geometry 2.3.1, scikit-learn 1.2.2, pandas 1.5.3, matplotlib 3.7.1.
This high-performance setup ensures efficient processing of large datasets and complex computations, providing a reliable foundation for experimental evaluation.

\paragraph{Dataset.}
The dataset consists of benchmarks widely adopted in logic synthesis, sourced from IWLS2005~\cite{IWLS2005} and IWLS2015~\cite{IWLS2015}. 
Table~\ref{tab:dataset} summarizes the characteristics of the dataset, including the number of primary inputs (\#PI), primary outputs (\#PO), gates (\#Gate), inverters (\#Inverter), two-input AND gates (\#AND2), edges (\#Edge), and circuit depth (Depth). 
The dataset’s variety—spanning small-scale designs like \textit{ctrl} to complex ones like \textit{sin}—enables robust evaluation of our approach.
This selection ensures diversity across the downstream tasks, offering a comprehensive testbed for \textit{BoolSkeleton}.

\paragraph{Baseline.}
Baseline methods are from a diverse set of techniques, including variation-based approaches~\cite{skeleton_variations} ('variation\_neighborhoods', 'variation\_edges', 'variation\_cliques'), edge-weight optimization~\cite{skeleton_heavy_edge_matching} ('heavy\_edge'), algebraic methods~\cite{skeleton_algebraic_distance} ('algebraic\_JC'), affinity-guided strategies~\cite{skeleton_affinity} ('affinity\_GS'), and Kronecker-based reduction~\cite{skeleton_kron_reduction} ('kron').
These methods, implemented using~\cite{pygsp, skeleton_variations}, are configured with a uniform reduction ratio of 0.3, and compared to assess their effectiveness in reducing graph complexity while preserving structural properties for downstream Graph Neural Network tasks.

\begin{table}[t]
\setlength{\tabcolsep}{3pt}
\centering
\scriptsize
\caption{The characteristics of the source designs.}
\label{tab:dataset}
\begin{tabular}{ l | r r r r r r r r r }
\toprule
\diagbox{\textbf{Design}}{\textbf{Char}}  & \#PI & \#PO & \#Gate & \#Inverter & \#AND2 & \#Edge & Depth    \\
\midrule
\textit{adder}	            & 256	& 129	& 2547	& 1527	& 1020	& 2172	& 511 \\
\textit{bar}	            & 135	& 128	& 6928	& 3592	& 3336	& 6928	& 21 \\
\textit{cavlc}	            & 10	& 11	& 1606	& 913	& 693	& 1400	& 33 \\
\textit{cht}	            & 47	& 36	& 595	& 324	& 271	& 614	& 17 \\
\textit{count}	            & 35	& 16	& 467	& 275	& 192	& 416	& 40 \\
\textit{ctrl}	            & 7	    & 26	& 419	& 245	& 174	& 380	& 20 \\
\textit{i2c}	            & 147	& 142	& 2785	& 1443	& 1342	& 2859	& 36 \\
\textit{int2float}	        & 11	& 7	    & 545	& 285	& 260	& 533	& 31 \\
\textit{max}	            & 512	& 130	& 6366	& 3501	& 2865	& 5988	& 495 \\
\textit{priority}	        & 128	& 8	    & 2350	& 1372	& 978	& 1970	& 499 \\
\textit{router}	            & 60	& 30	& 491	& 234	& 257	& 545	& 73 \\
\textit{s510}	            & 25	& 15	& 432	& 211	& 221	& 470	& 15 \\
\textit{sasc}	            & 113	& 63	& 692	& 337	& 355	& 809	& 16 \\
\textit{sin}	            & 24	& 25	& 11479	& 6063	& 5416	& 10879	& 350 \\
\textit{spi}	            & 240	& 239	& 8037	& 4250	& 3787	& 8049	& 63 \\
\textit{ss\_pcm}	        & 104	& 90	& 829	& 426	& 403	& 932	& 14 \\
\textit{steppermotordrive}	& 28	& 27	& 330	& 150	& 180	& 395	& 16 \\
\textit{ttt2}	            & 24	& 21	& 793	& 442	& 351	& 730	& 28 \\
\textit{unreg}	            & 36	& 16	& 411	& 253	& 158	& 348	& 16 \\
\textit{usb\_phy}	        & 132	& 90	& 963	& 476	& 487	& 1101	& 16 \\
\midrule
\textit{aes}	        & 683	& 529	& 51173	& 22518	& 28655	& 115418	& 44 \\
\textit{div}	        & 128	& 128	& 64826	& 37726	& 27100	& 108555	& 8406 \\
\textit{mem\_ctrl}	    & 1187	& 962	& 20324	& 10323	& 10001	& 41691	    & 58 \\
\textit{multiplier}	    & 128	& 128	& 58958	& 31205	& 27753	& 111242	& 524 \\
\textit{log2}	        & 32	& 32	& 68409	& 36027	& 32382	& 129590	& 597 \\
\textit{square}	        & 64	& 128	& 43595	& 24096	& 19499	& 78151	    & 445 \\
\textit{sqrt}	        & 128	& 64	& 78246	& 45647	& 32599	& 130518	& 10384 \\
\textit{vouter}	        & 1001	& 1	    & 24736	& 14208	& 10528	& 42114	    & 113 \\
\bottomrule
\end{tabular}
\vspace{-1em}
\end{table}

\subsection{Validity evaluation}

\subsubsection{Analysis 1: Boolean Network Compression}

\begin{figure}[t]
\centering
    \subfigure[Graph Size]{
    \includegraphics[width=0.225\textwidth]{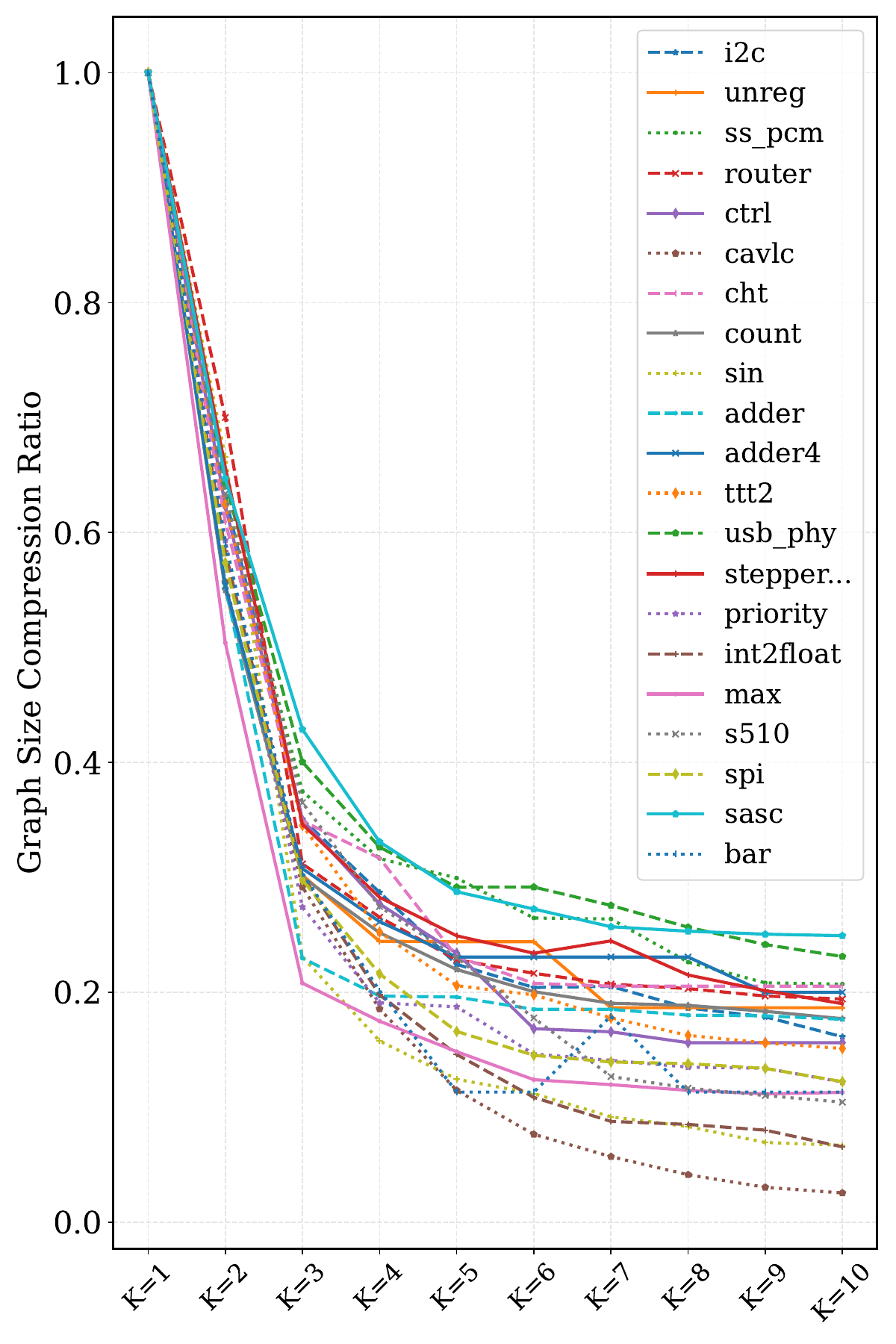}
    }
    \hfill
    \subfigure[Graph Depth]{
    \includegraphics[width=0.225\textwidth]{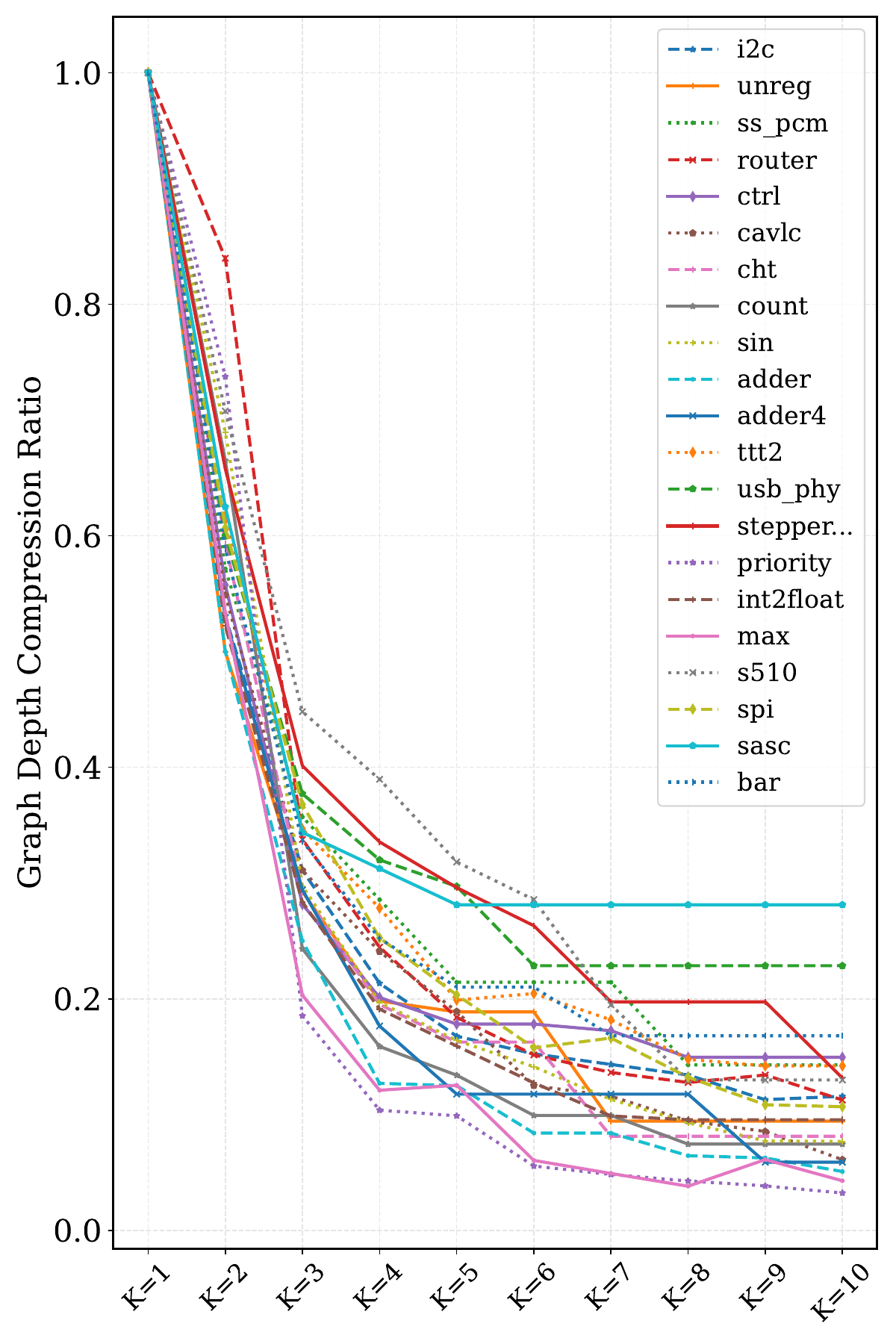}
    }
    \caption{Compression ratio over the graph size and depth.}
    \label{fig:task:compression:ratio}
    \vspace{-1em}
\end{figure}

\begin{figure}[t]
\centering
    \subfigure[Runtime of different Boolean network skeletons.]{
    \includegraphics[width=0.45\textwidth]{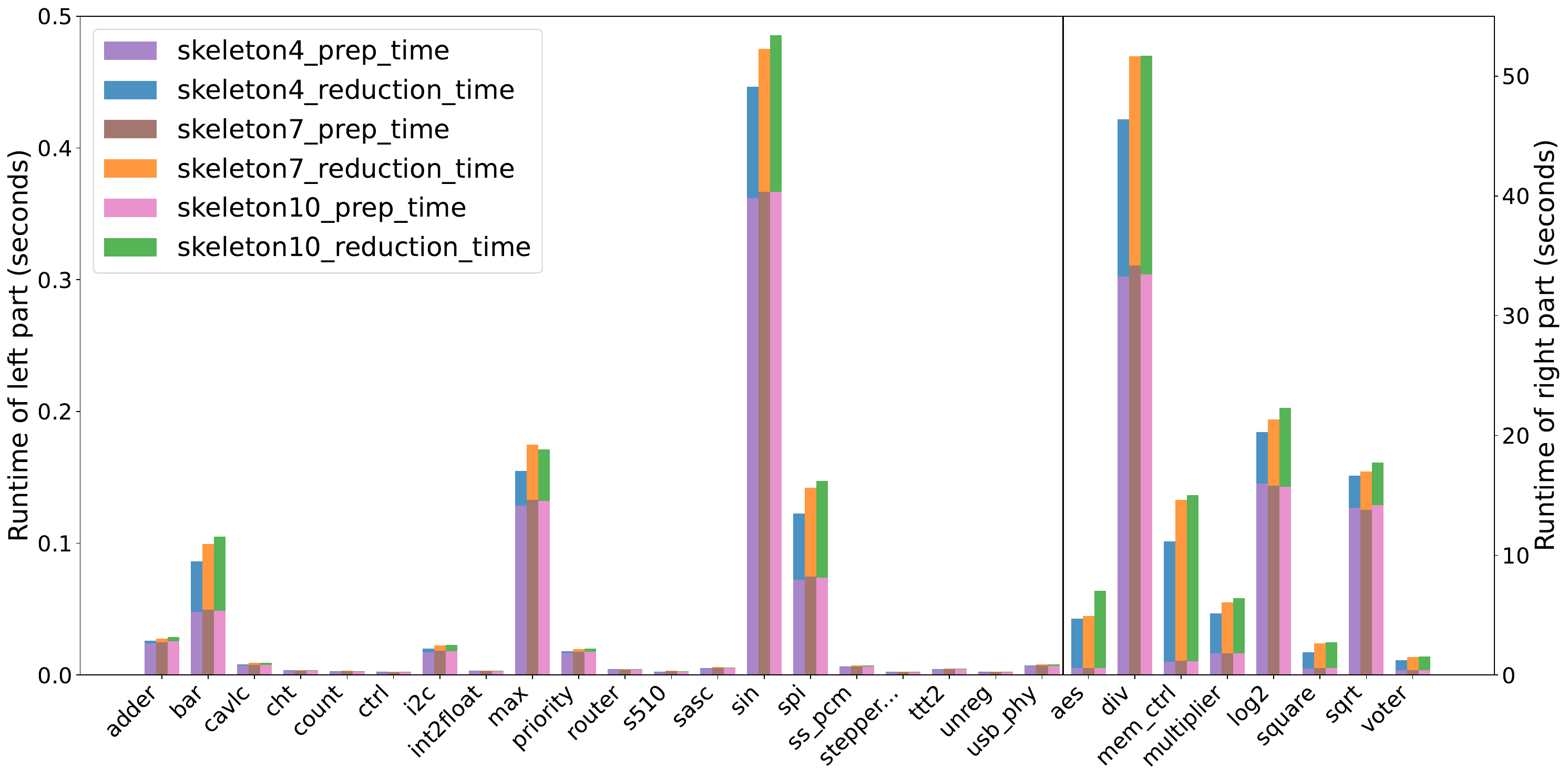}
    }
    \hfill
    \subfigure[Runtime comparison between different methods.]{
    \includegraphics[width=0.45\textwidth]{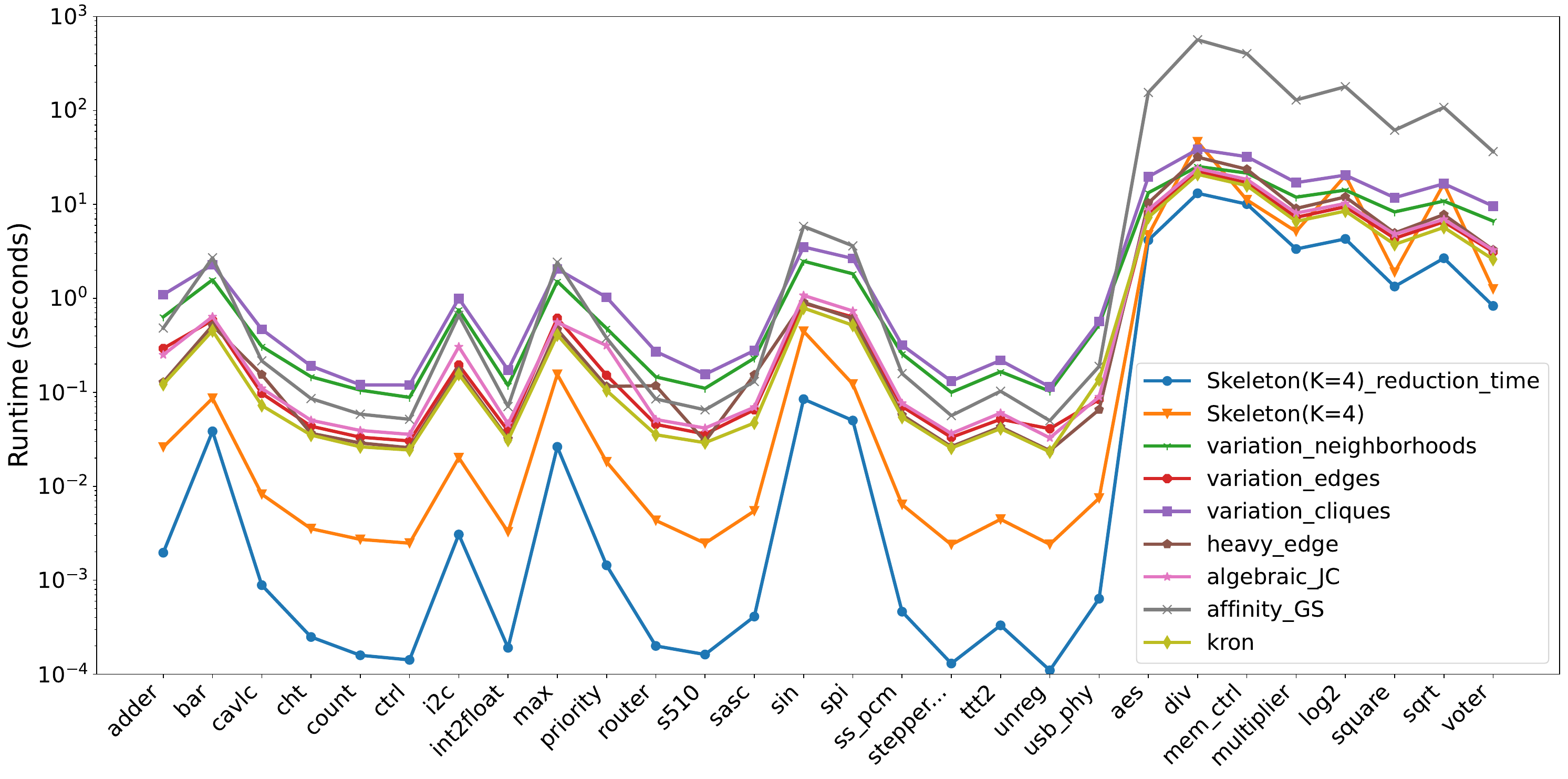}
    }
    \caption{Runtime analysis.}
    \label{fig:task:compression:runtime}
\end{figure}

\cref{fig:task:compression:ratio} illustrates the compression ratios achieved across selected designs with $K$ from 1 to 10;. 
This diagram reveals that increases in the fanin constraint \(K\) correspond to reductions in both graph size and depth. 
Specifically, a larger \(K\) permits greater merging of local nodes, enhancing the compression effect.

\cref{fig:task:compression:runtime}(a) illustrates the runtime performance of Boolean skeleton methods for $ K = 4, 7, $ and $ 10 $ across the selected designs. 
Each bar is composed of preparation time and node reduction time, highlighting that preparation time constitutes the predominant portion of the total runtime. 
Additionally, the runtime for a given design remains relatively consistent across the different $ K $ values.
\cref{fig:task:compression:runtime}(b) presents a runtime comparison between the Boolean skeleton methods with $ K = 4 $ and the baseline graph coarsening methods. 
This comparison demonstrates that the Boolean skeleton method consistently outperforms all baseline approaches in terms of runtime.

\subsubsection{Analysis 2: Boolean Network Classification}
\label{sec:task:classify}

The Boolean network classification task aims to evaluate the efficiency and robustness of the proposed skeleton method within the context of Boolean network analysis. 
The criterion for classification is that Boolean networks with the same functionality belong to the same class.
This task examines two primary aspects: (1) the extent to which the skeleton enhances computational efficiency for a specific Boolean network type, and (2) the method’s ability to maintain consistent classification performance across diverse Boolean network representations of a given design.

\paragraph{Dataset Generation Flow.}

\begin{figure}[t]
    \centering
    \includegraphics[width=1\linewidth]{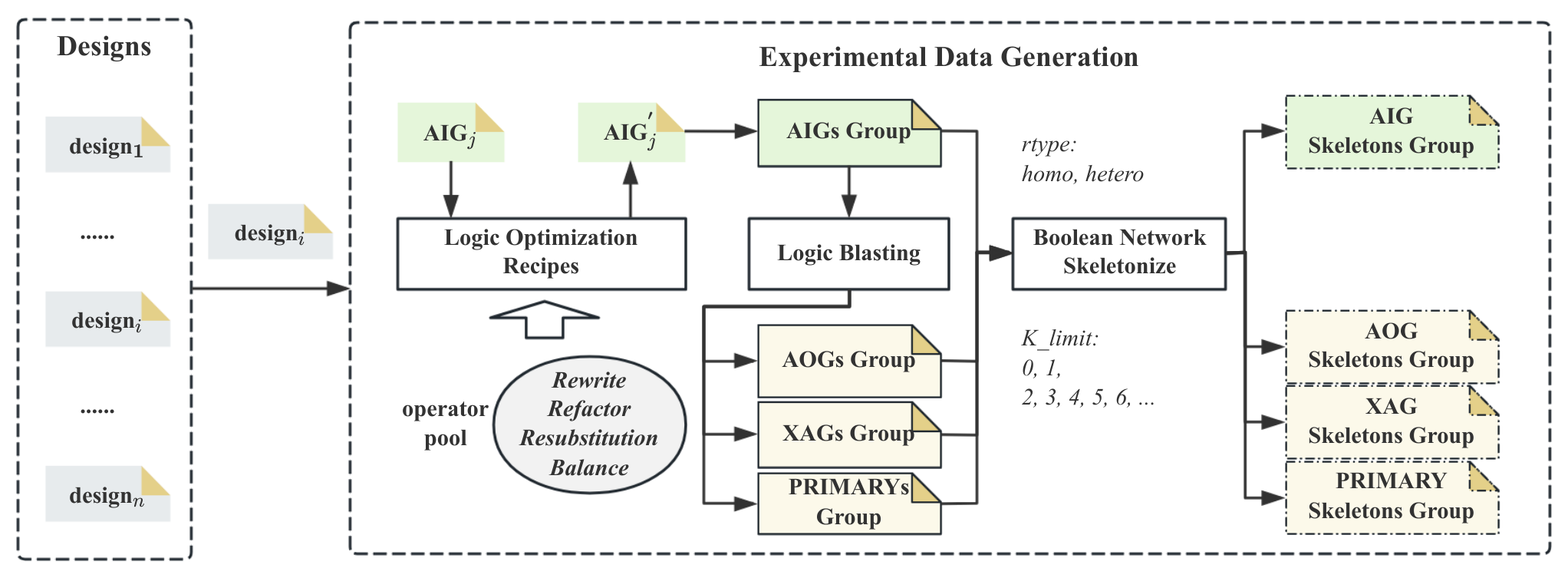}
    \caption{Boolean network classification dataset flow.}
    \label{fig:task:classify:datagen}
    \vspace{-1em}
\end{figure}

\cref{fig:task:classify:datagen} shows the Boolean network classification dataset generation flow.
For each design, the process begins by loading it into an And-Inverter Graph (AIG)-based Boolean network. 
Subsequently, logic optimization techniques are applied to produce a diverse set of AIG variants with a size of 1,000.
Each optimization sequence is randomly and redundantly generated by the operators in \{rewrite, refactor, balance, resub\} from berkeley-abc~\cite{BerkeleyABC} tool with a max size of 10.
These variants are then converted into other types of Boolean networks by logic blasting~\cite{openlsdgf}. 
Finally, the proposed Boolean network skeleton methods with different $K$ are employed to extract the corresponding skeleton from each network.

\paragraph{GCN-based Model.}

\begin{figure}[t]
    \centering
    \includegraphics[width=0.45\textwidth]{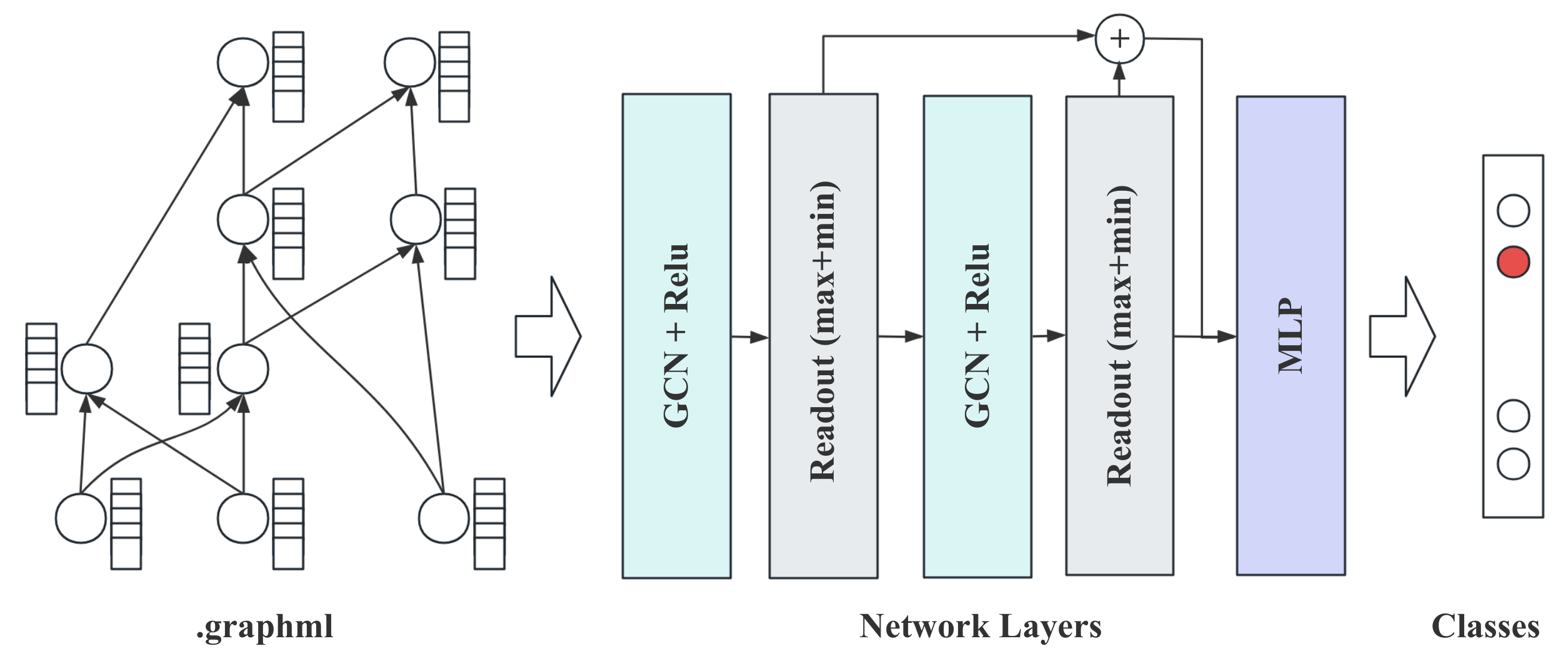}
    \caption{The GCN-based model of the graph classification.}
    \label{fig:classification_architecture}
    \vspace{-1em}
\end{figure}

\begin{figure}[t]
\centering
    \subfigure[Comparison between different methods over AIGs.]{
    \includegraphics[width=0.48\textwidth]{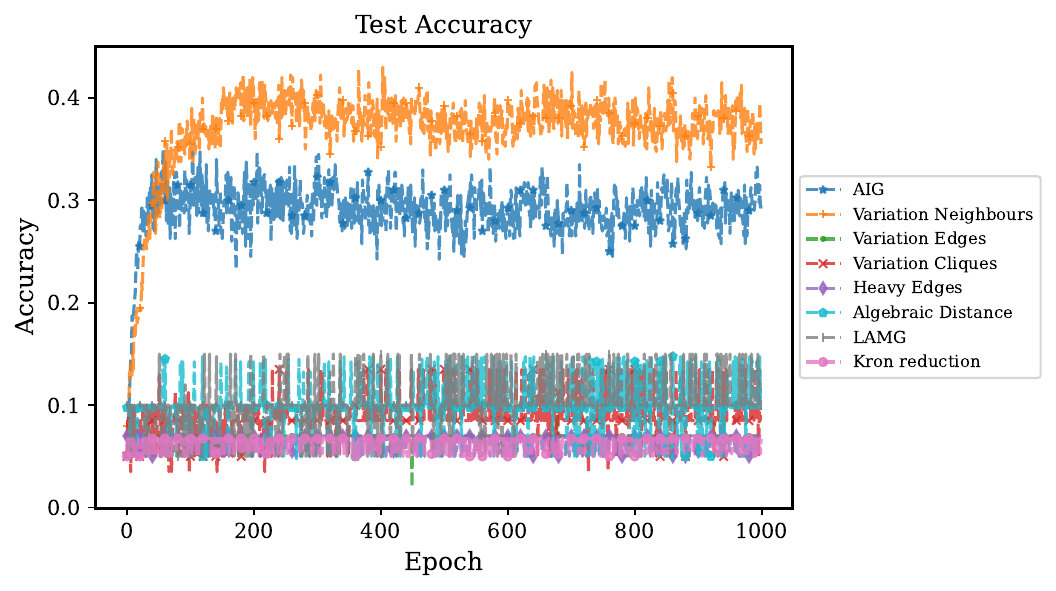}
    }
    \hfill
    \subfigure[Comparison between \textit{BoolSkeleton} with differnt $K$.]{
    \includegraphics[width=0.46\textwidth]{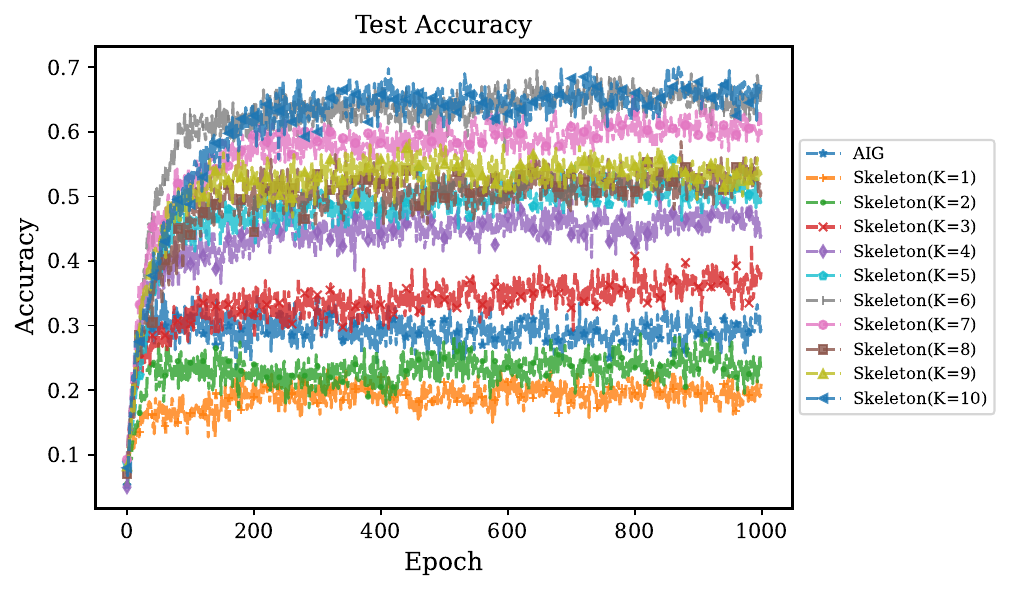}
    }
    \caption{Test accuracy comparison for classification task.}
    \label{fig:task:classify:exper:acc}
\end{figure}

The graph convolutional network (GCN) architecture employed for classification is depicted in \cref{fig:classification_architecture}. 
The input consists of DAGs in \textit{GraphML} format for the \textit{torch\_geometry}~\cite{pyg} package, augmented with node features representing circuit properties (node types by one-hot embedding). 
This model comprises \textit{``GCN + ReLU''} layers, followed by a graph embedding \textit{``Readout''} stage, and concludes with an \textit{``MLP + Softmax''} layer to predict its class label.
The training process spans 1,000 epochs, with a learning rate of 0.001, and employs cross-entropy loss as the loss function.

\paragraph{Evaluation.}

We conduct this experiment across two dimensions: (1) efficiency improvements for a specific Boolean network type, and (2) consistency across heterogeneous Boolean network representations.

\paragraph{Evaluation 1: The efficiency of the skeleton for one specific Boolean network type.}
\cref{fig:task:classify:exper:acc}(a) shows that the test accuracies of the compared baseline are worse than the original AIG dataset~(which means do nothing), excluding the ``Variation neighbors'' method.
\cref{fig:task:classify:exper:acc}(b) illustrates the test accuracy of AIG classification across different fanin constraints (\(K\)).
The results highlight how \textit{BoolSkelegon} improves generalization, with different \(K\) values striking a balance between complexity reduction and preservation of functionality-related features, as evidenced by elevated accuracy scores.

\begin{figure*}[h]
    \centering
    \subfigure[Skeleton(K=1)]{
    \includegraphics[width=0.322\textwidth]{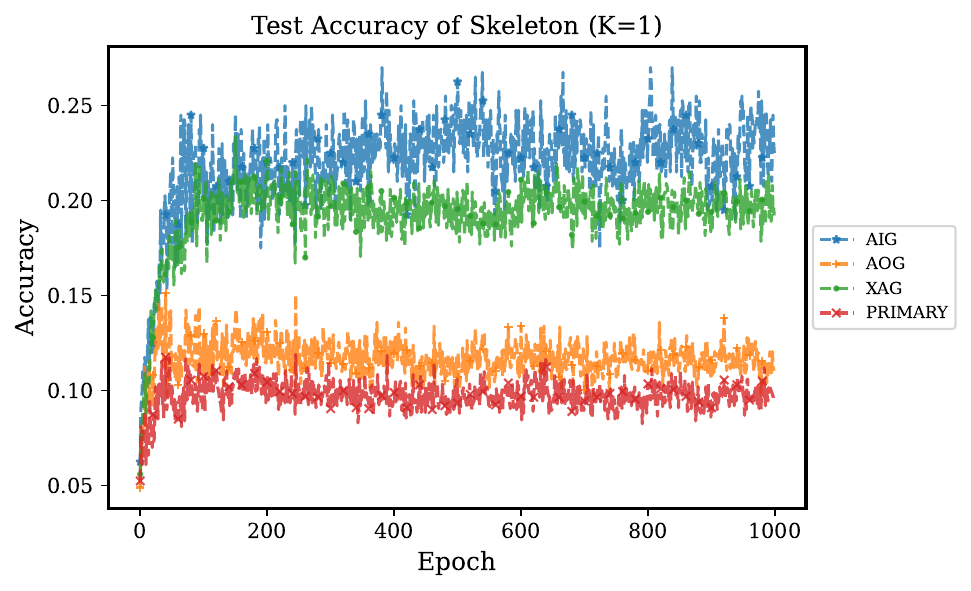}
    }
    \subfigure[Skeleton(K=4)]{
    \includegraphics[width=0.27\textwidth]{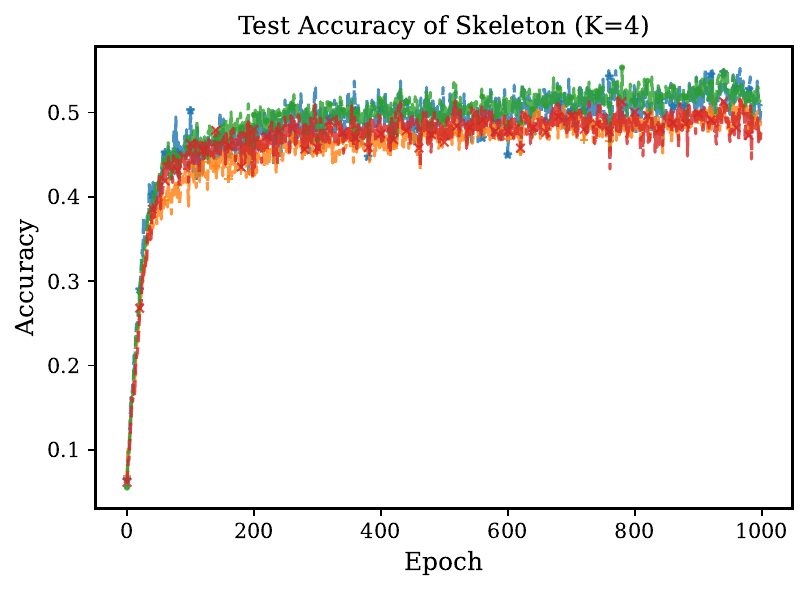}
    }
    \subfigure[Skeleton(K=8)]{
    \includegraphics[width=0.27\textwidth]{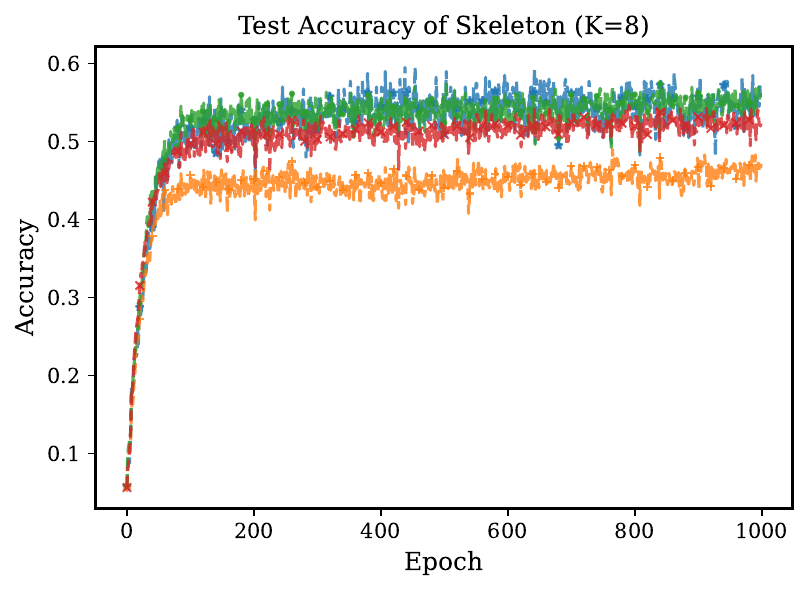}
    }
    \caption{Test accuracy comparison of the different Boolean network types while training with AIG.}
    \label{fig:task:classify:exper:consistency}
    \vspace{-1em}
\end{figure*}

\paragraph{Evaluation 2: The consistency of the skeleton across different types of Boolean network.}
This evaluation examines the robustness of the skeleton method across varied Boolean network representations, a critical factor for EDA applications where circuits may be expressed in multiple forms (AIG, AOG, XAG, PRIMARI\footnote{PRIMARY network consists of the basic primary logic gates.}). 
\cref{fig:task:classify:exper:consistency} compares classification accuracy for AIG-trained models applied to skeletonized networks with homogeneous pattern reductions and fanin limits ranging from \(K=1\) to \(K=10\). 
The consistent accuracy across these representations, despite structural differences, underscores the method’s adaptability. 
This resilience to variability enhances its potential for broad adoption in EDA workflows, where maintaining functional consistency across diverse circuit models is paramount.

\subsection{Task 1: Critical Path Analysis}
\label{sec:exper:task1}

In this task, we evaluate the similarity between the critical path\protect\footnote{The critical path is defined as the path with the maximum arrival time, as determined by the static timing analysis tool iSTA~\cite{iEDA}.} of the original Boolean network and that of its skeleton graph following technology mapping.

\begin{algorithm}[t]
\caption{Similarity Computation Between the Boolean Dependency Graph and Netlist}
\label{alg:similarity}
\begin{algorithmic}[1]
    \REQUIRE Boolean Network $\mathcal{C}^{logic}$, Boolean Dependency Graph $\mathcal{G}$, Gate-level Netlist $\mathcal{C}^{net}$
    \ENSURE The similarity $\alpha$
    \STATE $path_{critical} \gets \text{compute\_critical\_path}(\mathcal{G})$
    \STATE $pathes_{topk} \gets \text{compute\_topk\_timing\_path}(\mathcal{C}^{net}, k=3)$
    \STATE $region_{1} \gets \text{extract\_critical\_region}(\mathcal{C}^{logic}, [path_{critical}])$
    \STATE $region_{2} \gets \text{extract\_critical\_region}(\mathcal{C}^{logic}, pathes_{topk})$
    \STATE $\alpha \gets \text{compute\_similarity}(region_1, region_2)$
    \STATE {\bfseries Return:} $\alpha$
\end{algorithmic}
\end{algorithm}

\paragraph{Similarity Computation.}
\cref{alg:similarity} delineates the procedure for computing the similarity between the Boolean dependency graph and the gate-level netlist. 
The core idea is to map the critical path of the Boolean dependency graph and the top-\(3\) timing paths of the gate-level netlist onto the original Boolean network (lines 1–2). 
The critical regions associated with these paths are extracted using a two-stage labeling technique: (1) a top-down traversal assigns label \(a\), followed by (2) a bottom-up traversal assigns label \(b\) to the labeled nodes (lines 3–4). 
The similarity score \(\alpha\) is then calculated as the overlap between these two mapped regions (line 5). 
A higher \(\alpha\) indicates greater similarity between the critical paths of the skeleton graph and the netlist, suggesting that the \textit{BoolSkeleton} usually more effectively captures the critical path characteristics compared to the original Boolean network.

\begin{figure}[t]
    \centering
    \includegraphics[width=0.5\textwidth]{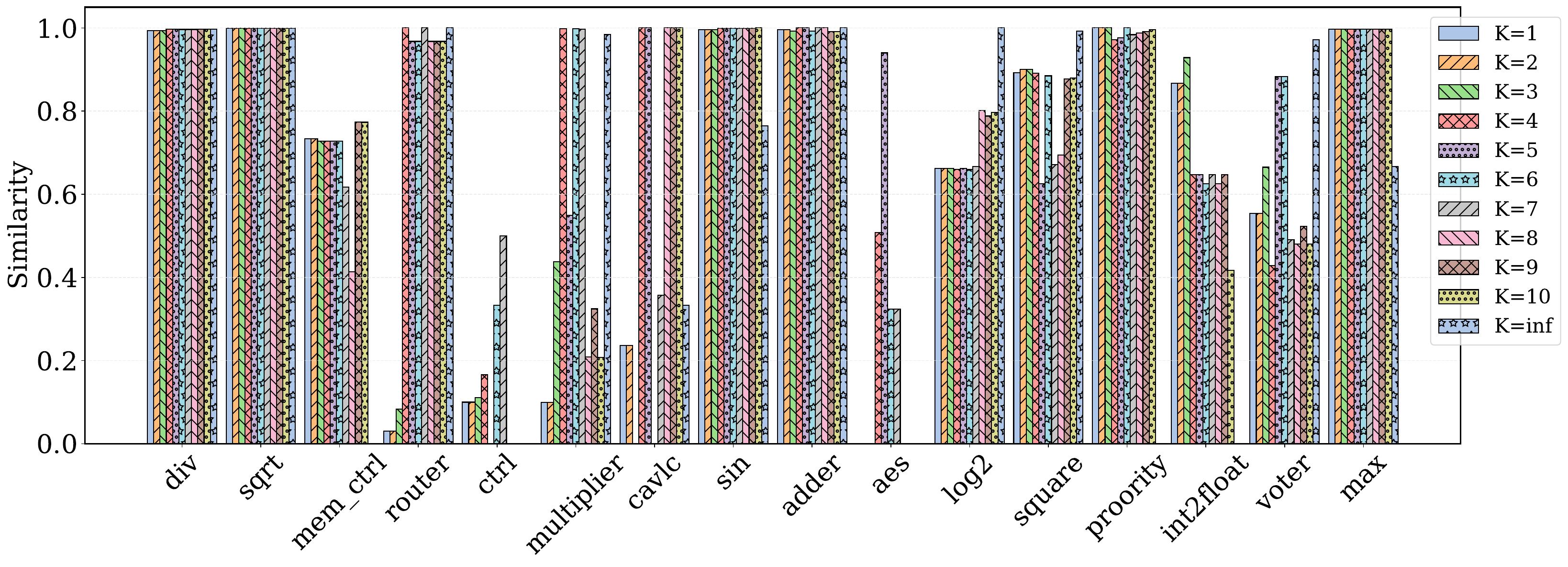}
    \caption{The similarity comparison between the skeleton graph with different constraints~($K$) and the gate-level netlist.}
    \label{fig:exper:task3:similarity}
    \vspace{-1em}
\end{figure}

\begin{figure}[t]
    \centering
    \includegraphics[width=0.5\textwidth]{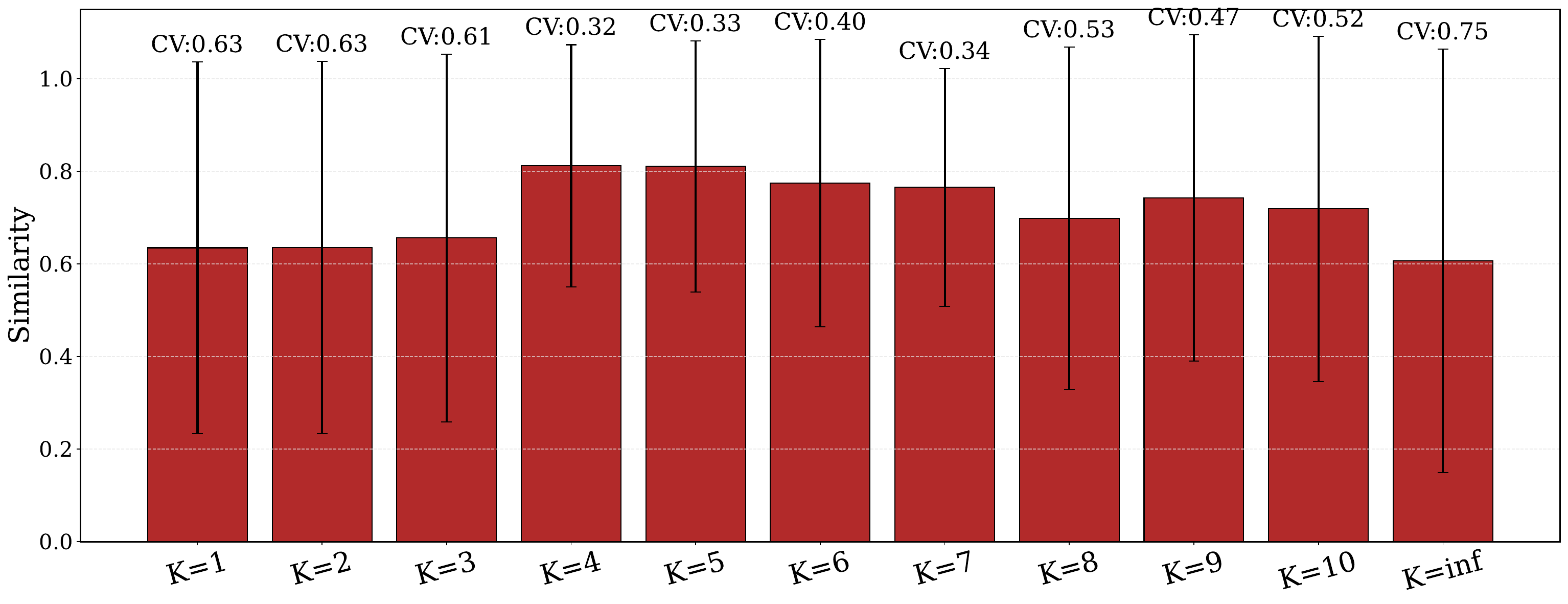}
    \caption{Comparison of average similarity (AVE) with standard deviation (STD, represented by error bars) and coefficient of variation (CV ($CV = \frac{AVE}{STD}$), values displayed above bars) across different methods.}
    \label{fig:exper:task3:error_bar}
    \vspace{-1em}
\end{figure}

\paragraph{Evaluation.}
\cref{fig:exper:task3:similarity} presents a similarity comparison between skeleton graphs, generated with varying fanin constraints (\(K\)), and the gate-level netlist for multiple designs. 
Most designs exhibit high similarity across different \(K\) values, indicating robust critical path preservation. 
However, certain designs—such as \textit{router}, \textit{ctrl}, and \textit{aes}—display varied similarity distributions, suggesting that critical path fidelity depends on an optimal balance of information retention. 
Excessive or insufficient skeletonization can degrade the representation of timing-critical paths.

According to the case study in \cref{sec:method:case}, the Boolean network classification task in \cref{sec:task:classify}, and the results in \cref{fig:exper:task3:similarity}, we infer the following: when the node count is excessively high, redundant paths emerge, complicating timing analysis; conversely, when the graph size is too low, significant information loss relative to the original circuit causes the critical path to deviate, increasing reliance on the original Boolean network for accurate timing prediction.

\cref{fig:exper:task3:error_bar} depicts the average similarity (AVE) with standard deviation (STD, represented by error bars) and coefficient of variation (CV, defined as \(CV = \frac{\text{STD}}{\text{AVE}}\), with values annotated above bars) across different \(K\) values for the skeleton method.
The results indicate that \(K=4\) and \(K=5\) yield superior performance across the test designs, balancing similarity and stability. 
From a cut enumeration perspective, a 4-feasible cut aligns closely with optimization and mapping strategies in technology mapping, explaining the method’s effectiveness.

\subsection{Task 2: Timing Prediction}
\label{sec:exper:task2}

\begin{figure}[t]
    \centering
    \includegraphics[width=0.5\textwidth]{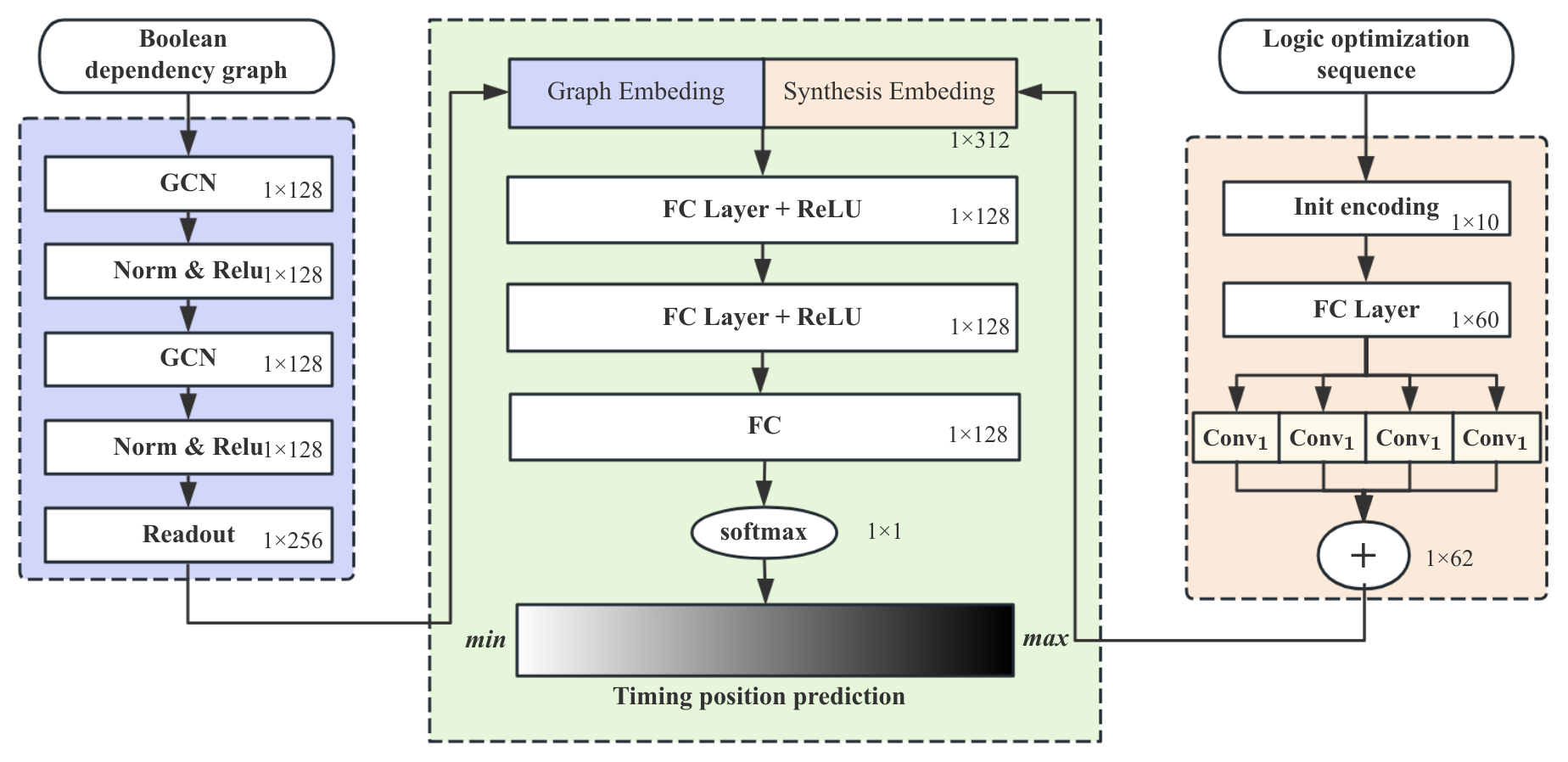}
    \caption{GCN-based timing prediction model.}
    \label{fig:task:timing_framework}
    \vspace{-2em}
\end{figure}

\begin{table*}[t]
    \centering
    \setlength{\tabcolsep}{1pt}
    \small
    \caption{The MAPE results comparison of the Timing Prediction task.}
    \begin{tabular}{ l  c  c c c  c c c c  c}
    \toprule
    \diagbox{\textbf{Design}}{\textbf{Graph}} & \textbf{AIG} & \makecell{ \textbf{Variation} \\ \textbf{neighbors} }$^{~\cite{skeleton_variations}}$  & \makecell{\textbf{Alberaic} \\ \textbf{Distance}}$^{\cite{skeleton_algebraic_distance}}$ & \makecell{\textbf{Kron} \\ \textbf{reduction}}$^{~\cite{skeleton_kron_reduction}}$  & \makecell{ \textbf{BoolSkeleton} \\(\textit{K=3})}  & \makecell{ \textbf{BoolSkeleton} \\(\textit{K=4})}  & \makecell{ \textbf{BoolSkeleton} \\(\textit{K=5})} &   \makecell{ \textbf{BoolSkeleton} \\(\textit{K=6})} & \textbf{STD.} \\
    \midrule
    \textit{adder}             & ~6.578    &  ~0.509   &  ~2.090   & ~0.668   &   ~0.524  &  \textbf{~0.004}   &  \textit{~0.351}   &  ~2.891   &  ~2.202  \\
    \textit{cavlc}             & ~1.612    &  \textit{~1.398}   &  ~5.155   & ~2.676   &   ~4.801  &  ~1.958   &  ~3.374   &  \textbf{~0.861}   &  ~1.590  \\
    \textit{cht}               & \textit{~4.937}    &  22.078  &  26.119  & ~7.290   &   10.533 &  \textbf{~3.014}   &  23.440  &  22.253  &  ~9.426  \\
    \textit{count}             & ~0.924    &  ~8.424   &  ~9.946   & ~4.601   &   \textit{~0.691}  &  ~1.780   &  ~1.262   &  \textbf{~0.091}   &  ~3.801  \\
    \textit{ctrl}              & ~5.524    &  ~2.343   &  15.126  & ~3.881   &   \textit{~1.887}  &  ~2.330   &  ~2.742   &  \textbf{~1.181}   &  ~4.544  \\
    \textit{i2c}               & ~1.882    &  ~3.868   &  ~1.620   & ~2.022   &   \textit{~0.053}  &  ~1.958   &  \textbf{~0.048}   &  ~1.947   &  ~1.219  \\
    \textit{int2float}         & ~2.407    &  ~3.153   &  \textit{~0.307}   & ~7.273   &   19.178 &  ~0.650   &  \textbf{~0.159}   &  ~1.521   &  ~6.428  \\
    \textit{max}               & ~0.597    &  ~2.603   &  ~2.734   & ~1.280   &   \textbf{~0.006}  &  ~0.158   &  ~0.152   &  \textit{~0.013}   &  ~1.145  \\
    \textit{priority}          & \textit{~1.204}    &  ~3.266   &  ~3.466   & 11.789  &   ~3.153  &  \textbf{~0.881}   &  ~2.268   &  ~3.405   &  ~3.429  \\
    \textit{router}            & ~5.011    &  \textbf{~0.167}   &  \textit{~0.642}   & ~6.195   &   14.029 &  ~9.426   &  ~8.220   &  ~2.243   &  ~4.758  \\
    \textit{s510}              & ~2.108    &  ~3.056   &  13.072  & ~4.244   &   16.001 &  ~0.483   &  \textbf{~0.011}   &  ~0.982                 &  ~6.098  \\
    \textit{sasc}              & ~0.635    &  ~2.465   &  12.354  & ~0.475   &   ~1.648  &  ~0.961   &  \textit{~0.436}   &  \textbf{~0.209}   &  ~4.092  \\
    \textit{sin}               & \textbf{~0.522}    &  ~2.483   &  \textit{~1.129}   & ~1.873   &   ~1.376  &  ~2.065   &  ~1.992   &  ~2.695   &  ~0.721  \\
    \textit{spi}               & 18.129    &  \textit{~0.587}   &  ~2.552   & ~5.424   &   ~2.358  &  ~4.915   &  \textbf{~0.033}   &  21.370  &  ~8.178  \\
    \textit{stepper}           & 15.392    &  14.050  &  14.324  & 12.001  &   \textbf{~1.400}  &  ~5.857   &  ~3.549   &  \textit{~2.997}   &  ~5.811  \\
    \textit{ttt2}              & 22.885    &  10.560  &  \textbf{~0.634}   & 38.760  &   12.759 &  ~4.346   &  \textit{~3.033}   &  31.543  &  14.070 \\
    \textit{unreg}             & ~1.908    &  ~7.607   &  33.813  & 13.733  &   ~1.506  &  \textbf{~0.456}   &  14.381  &  \textit{~0.645}   &  11.450 \\
    \textit{usb\_phy}          & ~8.607    &  \textbf{~1.991}   &  ~9.031   & ~4.384   &   ~8.428  &  ~3.817   &  \textit{~2.866}   &  ~5.115   &  ~2.782  \\
    \midrule
    \textbf{Train time (secs).}& 21896    &  1753   &  1262   & 1229   &  14498   &  12783   &   14659  &  14447  &  - \\
    \midrule
    \textbf{AVE.}              & 5.603     &  5.034   &  8.562   & 7.143   &  5.574   &  \textbf{2.503}   &  \textit{3.795}   &  5.665   &  - \\
    \textbf{TRIMAVE.}          & 4.841     &  4.273   &  7.500   & 5.583   &  5.072   &  \textbf{2.227}   &  \textit{2.804}   &  4.400   &  - \\
    \textbf{STD.}              & 6.619     &  5.662   &  9.411   & 8.812   &  6.221   &  \textbf{2.427}   &  \textit{6.061}   &  9.224   &  - \\
    \midrule
    \textbf{Impro~(AVE).}      & -        &  10.166   &  -52.797  & -27.470  &  ~0.526   &  \textbf{55.326}  &  \textit{32.267}  &  -1.091  &  - \\
    \textbf{Impro~(TRIMAVE).}  & -        &  11.738   &  -54.921  & -15.337  &  -4.767   &  \textbf{54.000}  &  \textit{42.075}  &  ~9.101   &  - \\
    \bottomrule
    \end{tabular}
    \label{tab:task:prediction}
\vspace{-1em}
\end{table*}

\begin{table}[tb]
    \centering
    \setlength{\tabcolsep}{2pt}
    \small
    \caption{ Performance comparison \nlwdel{between}\nlwadd{among} GCNs. }
    \begin{tabular}{ c c c c c }
        \toprule
        \multirow{2}{*}{\diagbox{\textbf{Metrics}}{\textbf{GCNs}}} & \multicolumn{2}{c}{\textbf{AIG}} & \multicolumn{2}{c}{\textbf{Skeleton(K=4)}} \\ \cmidrule(lr){2-3} \cmidrule(lr){4-5}
                                                  &  HuberLoss    &  TrainTime\nlwnew{(s)}    &  HuberLoss    &  TrainTime\nlwnew{(s)}  \\
        \midrule
        GIN$^{\cite{GIN}}$                        &  \textbf{0.252}    &  ~15829  &   \textbf{0.190}   &  12783        \\
        SAGE$^{\cite{SAGE}}$                      &  0.292    &  ~~7996  &   0.325   &  ~\textbf{7312}         \\
        HOGA$^{\cite{HOGA}}$                      &  0.430    &  183353  &   0.432   &  62041        \\
        BoolGebra$^{\cite{boolgebra}}$            &  0.422    &  ~~\textbf{6748}  &   0.253   &  13723       \\
        \bottomrule
    \end{tabular}
    \label{tab:exper:task4:gcn_comparison}
\end{table}

Quality of Results (QoR) prediction tasks~\cite{openabcd_animesh21, LOSTIN, MTLSO} are garnering growing attention, as they play a pivotal role in steering the optimization process within logic synthesis.
As highlighted by the critical path analysis tasks in \cref{sec:exper:task1}, the skeleton graph appears to offer a more effective representation of timing-critical paths compared to the original Boolean network. 
To further assess the efficacy of our approaches, we undertake a timing prediction task.

\paragraph{Dataset.}
The timing prediction dataset is constructed based on the classification dataset above. 
\nlwnew{We use ``map'' command of ``ABC''~\cite{BerkeleyABC} tool to perform technology mapping, and ``sky130'' is used as the standard cell library.}
Building on this foundation, the static timing analysis tool, iSTA~\cite{iEDA}, was employed to assess the "arrival time" of the netlist generated for each design across its respective optimization sequence. 

\paragraph{GCN-Based Model.}
\cref{fig:task:timing_framework} illustrates the GCN-based timing prediction model tailored for logic optimization. 
Training data is generated per the dataset flow in \cref{fig:task:classify:datagen}, with each sample comprising a Boolean dependency graph and its associated logic optimization sequence. 
The model aims to predict the timing quality of a Boolean network under a given optimization action. 
It consists of three core components: (1) a graph embedding module, employing a two-layer GCN for node feature aggregation followed by a readout layer (mean + max) to produce a graph-level feature vector; (2) a synthesis embedding module, encoding the discrete optimization sequence using four convolutional filters with kernel sizes \{1×14, 1×15, 1×16, 1×17\}; and (3) a timing quality prediction module, utilizing fully connected layers and a softmax activation to estimate the timing quality as a probability distribution across normalized timing range.
The training process spans 500 epochs, with a learning rate of 0.001, and employs Huberloss~\cite{huberloss_math1964_math} as the loss function.

\paragraph{Evaluation.}
\cref{tab:exper:task4:gcn_comparison} shows performance comparison results of various GCNs applied to two graph types: the original AIG and the skeleton graph by the proposed \textit{BoolSkeleton} method with $K=4$.
It suggests that the GIN-based timing prediction model surpasses all the other models in the loss evaluation, although the training time is comparatively higher.

\cref{tab:task:prediction} reports the Mean Absolute Percentage Error (MAPE) for timing prediction across various graph-based methods applied to the dataset in \cref{tab:dataset}, with lower MAPE values signifying higher accuracy. 
Among the evaluated approaches, \textit{BoolSkeleton} with \(K=4\) consistently outperforms others, achieving the lowest average MAPE (AVE: 2.503) and trimmed mean (TRIMAVE\protect\footnote{TRIMAVE is the average excluding the maximum and minimum values.}: 2.227), surpassing the baseline AIG by 55.3\% and 54.0\%, respectively. 
In contrast, methods like \textit{Alberaic Distance} (AVE: 8.562) and \textit{Kron reduction} (AVE: 7.143) exhibit higher errors, underperforming AIG by 52.797\% and 27.470\%, respectively. 
The standard deviation (STD) across designs further underscores the stability of \textit{BoolSkeleton} (\(K=4\): 2.427, lowest), while per-design STD (rightmost column) reveals circuit-specific variability, with designs like \textit{ttt2} (14.070) showing greater sensitivity to skeletonization than \textit{sin} (0.721). 
These results suggest that \textit{BoolSkeleton} effectively balances accuracy and robustness, offering a superior approach for timing prediction.

\section{Discussion}
\label{sec:discussion}
In this section, we will discuss the advantages and limitations of the proposed \textit{BoolSkeleton}.

\paragraph{Advantages.}
\textit{BoolSkeleton} is designed from the functionality viewpoint, enabling it to distill a consistent structural representation across diverse Boolean network variants. 
This method significantly enhances performance in functionality-related tasks of the Boolean networks by preserving essential Boolean dependency properties within structural variability.
Additionally, the node-level fanin-limited homogeneous pattern reduction process efficiently extracts a coarse-grained skeleton from the Boolean network, facilitating the capture of high-dimensional representations.
The parameter $K$ of the pattern reduction operator can also control the coarsening ratio of the Boolean dependency graph.
Experimental results validate these advantages, showcasing superior outcomes across multiple downstream applications.
Furthermore, it holds the potential to improve functionality- and profile-related tasks, for example, it can help the matches between ports in Boolean matching~\cite{matching_date10_npnp}, and it can enhance the optimization space exploration by PPA prediction~\cite{rtlppa_iccad2022_ppa, rtldelay_iccad2022_ppa, masterrtl_iccad23_ppa}, etc.

\paragraph{Limitations}
There is a reduction ratio for the graph coarsening method, and the parameter $K$ of the pattern reduction operator is used to control the reduction ratio in $BoolSkeleton$.
While this flexibility of $K$ also expands the scope of analytical tasks, determining an optimal \(K\) value for specific Boolean networks and tasks poses a challenge.
Consequently, an adaptive $BoolSkeleton$ tailored to individual tasks is also essential.
For large-scale circuits, a potential solution is to combine with ``BoolSkeleton'' through graph partitioning algorithms, which is also consistent with the nature of large circuits that are generally composed of sub-modules.


\section{Conclusion}
\label{sec:conclusion}

In this work, we proposed \textit{BoolSkeleton} to address the challenge of the coarse-grained Boolean network representation in logic synthesis.
\textit{BoolSkeleton} is tailored to balance their static functionality and dynamic structural variability.
Through an in-depth analysis, we identified critical attributes: Boolean dependency, reachability, reconvergence, and the conflict between static functionality and dynamic DAG. 
Leveraging these insights, \textit{BoolSkeleton} employs a preprocessing phase to transform Boolean networks into dependency graphs, followed by an iterative node-level homogeneous pattern reduction process. 
This approach preserves coarse-grained functionality-related information while simplifying fine-grained topological structures, addressing the limitations of existing GNN-based methods that overly rely on local node embeddings.
The experimental results of four downstream tasks demonstrate the efficacy of \textit{BoolSkeleton}.
Future work will focus on the integration with advanced machine learning platforms or frameworks to \textit{BoolSkeleton}'s applicability across a broader range of Boolean network representation tasks.

\bibliographystyle{IEEEtran}
\bibliography{skeleton}

\vfill

\end{CJK}
\end{document}